\documentclass{elsarticle}
\usepackage{natbib}
\usepackage{graphicx}
\usepackage{xcolor}
\usepackage{amsthm}
\usepackage[linesnumbered,ruled]{algorithm2e}
\usepackage[margin=1in]{geometry}
\usepackage{hyperref}

\usepackage{subcaption}
\usepackage{amsmath}
\usepackage{enumitem}
\usepackage{algpseudocode}
\usepackage{rotating}
\usepackage{comment}
\usepackage{multirow} 
\newtheorem{theorem}{Proposition}

\begin{document}
\begin{frontmatter}
\title{An Exact Solution Algorithm for the Bi-Level Optimization Problem of Electric Vehicles Charging Station Placement}
\author{
  Mobina Nankali\textsuperscript{a}, 
  Michael W. Levin\textsuperscript{a,}\footnote{Corresponding author: mlevin@umn.edu}\\
  \textsuperscript{a}Department of Civil, Environmental, and Geo- Engineering,\\University of Minnesota,\\
  500 Pillsbury Drive SE, Minneapolis, MN 55455, USA.\\
  Emails: \texttt{nanka005@umn.edu}, \texttt{mlevin@umn.edu}
}
\begin{abstract}
This work addresses electric vehicle (EV) charging station placement through a bi-level optimization model, where the upper-level planner maximizes charging revenue by selecting station locations under budget constraints, while EV users at the lower level choose routes and charging stations to minimize travel and charging costs. To account for battery range limitation, we construct a battery-expanded network and apply a shortest path algorithm with Frank-Wolfe traffic assignment.
Our primary contribution is developing the first global solution algorithm for large scale EV charging station placement problems. We propose a Branch-and-Price-and-Cut algorithm enhanced with value function cuts and column generation. Our exact algorithm delivers solutions provably within the reported optimality gap of the global optimum with mathematical certainty.
Computational experiments on the Eastern Massachusetts network (74 nodes, 248 links), the Anaheim network (416 nodes, 914 links), and the Barcelona network (110 zones, 1,020 nodes, and 2,512 links) demonstrate exceptional performance. Our algorithm terminates in a reasonable runtime while achieving optimality gaps below 1\% across all instances. Controlled benchmarks against two genetic algorithms on identical instances confirm that the proposed algorithm finds equal or better solutions in less computation time across all tested networks. The algorithm successfully handles problems with over 300,000 feasible combinations, transforming EV charging infrastructure planning into a tractable optimization suitable for practical decision making on real-world networks with optimality guaranteed.
\end{abstract}

\end{frontmatter}
\newpage
\section{Introduction}
The shifting trend to EVs has created a need for efficient, accessible, and sustainable charging infrastructure. Strategic placement of EV charging stations is critical to accommodate growing demand and to shape user behavior and ensure equitable distribution of services across a network \cite{ESMAILI2024105409}. However, the planning problem is inherently complex due to the hierarchical interaction between infrastructure providers and EV users. 

Determining optimal locations for EV charging stations is a critical infrastructure planning challenge that depends on multiple interdependent factors, including user behavior, network flow patterns, demand distribution, and the broader energy ecosystem. Beyond serving transportation needs, strategically located stations can also reduce the economic and safety concerns associated with fuel transportation and being resilient \cite{GHORBANALIZADEGAN2026104627, movahedi}. Since building charging stations at every site with potential demand is infeasible due to budget constraints, planners must strategically identify locations that maximize accessibility and efficiency of usage. However, the optimal placement of charging stations is not simply a matter of available budget. It is also closely tied to user responses and the resulting traffic flow within the transportation network. Naturally, users try to minimize their travel and charging costs, meanwhile their collective choices influence network utilization and congestion. This paper addresses the charging station deployment problem using a bi-level optimization framework. At the upper level (UL), a network planner seeks to design a profitable charging network by selecting station locations within a fixed budget and estimating revenue from user charging behavior. Moreover, EV users act as rational agents that minimize travel and charging costs, which depend on the network configuration. The interplay between these two decision making layers requires a two-level model that captures both planner objectives and user equilibrium (UE) responses.

Many papers have explored the charging station placement for EVs using a bi-level framework \cite{doi:10.1287/trsc.2021.0494, 10122631, 10034431, 9690622, 9944904}. \citet{9944904} plan urban roads, EV charging stations, and the power grid by a bi-level model formulation where the UL is investment decisions and the lower level (LL) is EV UE routing. They convert the entire model to a single mixed integer linear program (MILP) via KKT conditions and an improved Big-$M$ linearization, then solve it exactly with a solver. Their formulation relies on very large $M$ constants. Large $M$ MILPs can still suffer weak relaxations and long solve times on bigger networks. They have run their method on a 41 nodes/56-bus case study; the authors note that adding expressways or finer zoning could induce larger computational burdens, and they suggest decomposition (e.g., Benders) or scenario reduction would be needed for bigger instances, but these techniques are not yet implemented. 

In another study, \citet{https://doi.org/10.1155/2017/4252946} formulates a bi-level problem where the UL planner selects some roadway link candidates to host chargers to maximize battery electric vehicle (BEV) flow, while a LL computes a stochastic UE restricted by path distance in which BEVs are limited by driving range. Based on the paper, the resulting mixed integer nonlinear problem (MINP) is intractable for exact solvers, and the authors propose a heuristic loop. They demonstrate the method on two benchmark road networks: Nguyen-Dupuis (13 nodes, 19 links, 4 O-D pairs) and Sioux Falls (24 nodes, 76 links, 576 O-D pairs). While their algorithm is fast on small graphs, heuristic algorithms do not offer an optimality guarantee. It can stall at a local maximum of covered BEV flow and provides no bound on how close that solution is to the true optimum. Similarly, \citet{Tran10122021} used a purely heuristic called Cross-Entropy Method. UL chooses nodes to minimize the sum of the stations' installation costs and the value of time spent traveling. At the same time, LL assigns mixed EV traffic via a multi class UE that respects EV range. It is shown that the algorithm converges in their test networks, but it is not guaranteed to converge to the global optimum. %The algorithm provides no optimality gap or confidence interval.

In contrast to these earlier studies, our work provides the first exact solution methodology that scales to metropolitan scale road networks. We formulate the charging station placement problem as a MINP identical in fidelity to the network models above, but we solve it directly without surrogate objectives, path enumeration, or heuristic sampling. The resulting algorithm guarantees that the returned solution is within the reported optimality gap of the global optimum, including for a network with over 1,000 nodes solved within hours. A controlled comparison against genetic algorithm benchmarks on identical instances confirms the runtime and solution quality advantages. Hence, our contribution is twofold: our solution algorithm guarantees the solution to be within the reported optimality gap of the global optimal, compared to the heuristics with no guarantee, and we demonstrate tractability on very large, real world networks. Building on insights from recent literature, the UL decision involves binary location variables and station cost constraints, while the LL subproblem models traffic assignment and charging flows using a convex UE formulation. We propose a Branch-and-Price-and-Cut (BPC) algorithm enhanced with value function cuts and column generation (CG), built upon the high point relaxation (HPR), to solve the bi-level problem to global optimality. This formulation enables network planners to anticipate user behavior and optimize system design in a way that balances profitability and user accessibility.

The remainder of this paper is organized as follows. Section~\ref{literature review} provides a literature review, covering planning approaches to EV charging station placement. Section~\ref{Methodology} presents our bi-level mathematical formulation. Sections~\ref{Solution Algorithm} and \ref{BPC} introduce the exact solution algorithm. Section~\ref{Numerical Results} shows the computational experiments on three networks ranging from 74 to 1020 nodes, and Section~\ref{sec:sensi} presents the sensitivity analysis. Finally, Section~\ref{Conclusion} concludes the paper result and discusses directions for future work.

\section{Literature Review}\label{literature review}
Network design problems are a foundational topic in transportation research, which involves decisions about how to modify a transportation network to improve overall system performance. These problems are often modeled as bi-level programs, where a system planner (the leader) makes infrastructure decisions (e.g., adding new links), and travelers (the followers) respond by selecting routes under UE. Depending on the nature of the planner's decisions, network design problems are classified as either continuous (e.g., adjusting capacities or tolls) or discrete (e.g., selecting which links to build or upgrade). In this literature review, we focus on the discrete network design problem (DNDP), where the planner chooses from a set of binary investment options, such as whether to install a facility or not, since our problem is discrete.
\subsection{Discrete Traffic Network Design Problem}
Exact solution methods for the DNDP have evolved primarily through progressively stronger lower-bounding formulations. One of the earliest contributions was made by \citet{leblanc1975algorithm}, who formulated link construction decisions using binary variables subject to a budget constraint, with the objective of reducing network congestion. To obtain valid lower bounds, LeBlanc introduced a  HPR in which the UE condition is replaced by a system-optimal (SO) traffic assignment. The relaxation further assumes that all unfixed candidate links are constructed and temporarily omits the budget constraint. Although computationally convenient, these optimistic assumptions can produce weak lower bounds.

An alternative early bounding strategy was proposed by \citet{GAO2005479}, who employed the support function concept from generalized Benders decomposition. Their approach shows the UL objective as a function of the binary link design variables and represents the bi-level problem through an iteratively refined single level master problem. At each iteration, the UE follower problem is solved using an augmented Frank Wolfe procedure to obtain both equilibrium flows and the Lagrange multipliers associated with the big-$M$ linking constraints. These multipliers are then used to construct a support function cut for the relaxed master problem. The resulting procedure generates a nondecreasing sequence of lower bounds and has finite convergence properties, although its computational evaluation was limited to relatively small networks, such as Sioux Falls.

To strengthen the HPR, \citet{Farvaresh2013} retained the budget constraint and the binary status of undecided links within the bounding problem. They developed a branch-and-bound (B\&B) framework with additional fathoming procedures and formulated each relaxation as a mixed-integer convex program solved through outer approximation (OA) iterations. By avoiding the assumption that all unfixed links are constructed, their method produces tighter lower bounds than the original HPR. They also transformed the path-based formulation into a link--node multicommodity representation, thereby eliminating path enumeration and allowing standard traffic assignment algorithms to be used more directly.

The DNDP was subsequently generalized by \citet{WANG201342} to include multiple capacity levels, so that each candidate link may receive zero, one, or several capacity additions. Their methods preserve the observation that SO routing provides an admissible lower bound while developing stronger global optimization procedures. The first method orders candidate designs according to their SO objective values and then evaluates them successively under UE conditions. The second method introduces the Beckmann--McGuire--Winsten objective of the UE problem as a convex inequality within the SO formulation, which reduced the relaxed feasible region without eliminating any bi-level feasible solutions.

In addition, \citet{WANG201342} jointly model link addition and link capacity decisions. This model allow the fixed capacity DNDP and the continuous network design problem to be represented as special cases. They express the UE condition in LL as a variational inequality and reformulate it using the corner points of the link-flow polytope. Bilinear construction cost and flow linking terms are handled through  reformulation and linearization technique. The nonlinear Bureau of Public Roads (BPR) travel-time functions are treated through OA after a logarithmic transformation, and a range reduction procedure progressively tightens the resulting MILP relaxation. This procedure is designed to converge to the global optimum of the original formulation rather than merely to the optimum of a fixed linear approximation.

A decomposition alternative was developed by \citet{https://doi.org/10.1111/mice.12224}. Instead of using the OA lower-bound formulation of \citet{Farvaresh2013}, they embed a Benders decomposition within the B\&B tree. At each node, the SO relaxation is decomposed into a convex traffic-assignment subproblem and a mixed-integer linear master problem. Dual information obtained from the subproblem is used to generate Benders cuts, producing tighter node level bounds. This framework enabled the exact treatment of a real-scale multimodal and multiclass transportation network.

More recent studies have combined combinatorial optimization with machine learning surrogates for simulating network design variants. \citet{Yin2022}, for example, embed B\&B within a Bayesian-optimization framework. A Gaussian model is fitted to the outputs of a simulated multimodal equilibrium, while B\&B is used to optimize the acquisition function over the integer and budget-constrained design space. In this setting, however, B\&B certifies optimality only for the surrogate acquisition problem. Because the underlying equilibrium remains a black-box simulation, the method does not provide a global-optimality guarantee for the original bi-level DNDP.

Most recently, \citet{ReyLevin2024} reformulate the SO relaxation in path space and exploit two important structural properties. First, the link-separable structure of total system travel time permits OA cuts to be generated independently for individual links. Second, the exponentially large set of feasible paths can be handled implicitly through CG, avoiding complete path enumeration. These ideas are used into a single-tree branch-price-and-cut (BPC) algorithm. Lower bounds are obtained from an evolving restricted-master linear program, while upper bounds are generated by solving the UE problem whenever the design variables at a node define a budget-feasible integer solution. The path-based master problem avoids multicommodity link-flow constraints and generates OA cuts only where they are required. The BPC framework obtains its lower bounds from a linear restricted-master problem. This distinction reduces the computational cost of each bounding evaluation and provides the principal algorithmic foundation for the approach developed in this study.

\subsection{EV charging station problem}
Many papers have considered the problem of choosing locations for the charging station for EVs without considering the response of the EV users \cite{doi:10.3141/2252-12, 10406825, 8388731, 10703145, 9713435, 202601.1754, ghamami2025evcharging}. For example, \citet{doi:10.3141/2252-12} estimated the day and night charge demand of EVs through census and employment regressions and solved a mixed integer problem (MIP) capable of delivering maximum coverage to site chargers for maximum coverage within a certain radius. However, by treating drivers as fixed demand points and omitting any UE feedback, the model cannot capture route choice and EV user response to station placement. \citet{10406825} solved EV charger siting and sizing under uncertain range and demand as a two stage stochastic MILP on 1,079 Pennsylvania demand nodes. The model treats vehicles as fixed demands with fictitious travel penalties, so it neglects route choice, congestion and other UE responses that a bi-level formulation would capture.

Some papers have considered a bi-level formulation and capturing the user's response \cite{doi:10.1287/trsc.2021.0494, 10122631, 10034431, 9690622, 9387405, 10196503, https://doi.org/10.1155/2017/4252946, HE2015227, 10422114, 8064175}. They can be divided into two categories, on whether they propose exact or heuristic solution algorithms. \citet{doi:10.1287/trsc.2021.0494} addressed strategic EV charging station location and sizing under stochastic flows and congestion using bi-level optimization. The UL minimized infrastructure costs while ensuring probabilistic service level constraints, while the LL represented EV users minimizing route lengths. They simplified the bi-level model into a single-level MILP with M/M/c queuing and developed an exact Benders' decomposition solution. However, the algorithm only guarantees global optimality on small networks. For larger cases, the MILP becomes intractable and results may be far from optimal with unknown gaps. \citet{10122631} proposed a bi-level model for EV charging station placement and direct current fast charger in highway networks, with the UL minimizing total annualized costs and the LL addressing UE traffic assignment and power distribution network operations. They converted this to a single-level MILP using KKT conditions, ensuring global optimality. Testing was limited to a modest 3-city network with 6 OD pairs and 14 pregenerated paths, leaving practical tractability on larger instances with thousands of potential paths unclear.

\citet{10034431} developed a hierarchical EV charging network design model with dynamic pricing. Their problem minimize facility and operating costs while maximizing operator revenue balanced against user travel and charging expenses. They reformulated the bi-level problem into a single-level model using KKT conditions and solved it with an iterative active set heuristic alternating location fixes with Macroscopic-Fundamental-Diagram (MFD) updates for travel times. The results on a network with 12 nodes showed a gap within 4.5\% of the exact solution, while on a 42 nodes campus network they achieved a 0.43\% gap. Without global optimality guarantees, this solution is a near-optimal heuristic rather than an exact solver. In another study, \citet{9690622} developed a hierarchical model for EV charging infrastructure design and operational management. The UL minimize deployment costs and maximize revenue through dynamic pricing, and the LL minimize user travel and charging expenses UE conditions. They implemented Mitsos's global optimization bounding scheme, which theoretically reaches global optimum but practically stopped at 4.58\% and 2.89\% gaps on test cases. This yields a solution that is near optimal. The algorithm's long run time on a network under 350 links shows intractability for larger networks without additional strategies.

\citet{HE2015227} framed public charger siting as a bi-level problem with budget constraint. The UL chooses station locations, and the LL uses tour-based equilibrium for EV drivers, which selects routes and recharge stops based on battery range and risk preferences. They reformulated this as a path-based program solved with a genetic algorithm (GA) coupled to CG for the equilibrium subproblem. Since the GA is heuristic without optimality guarantees and testing was limited to the modest Sioux Falls network (24 nodes, 76 links), scalability and solution quality on larger networks remain uncertain.
\citet{10422114} proposed a bi-level model for electric autonomous vehicles with an UL modified p-median locating fixed fast charging stations to minimize flow weighted access distance, and a LL SO traffic assignment with stochastic charging demand minimizing total travel, charging, and queuing costs. They solved this with an iterative heuristic alternating traffic re-assignment and station re-siting, which relax the nonlinear LL via generalized Lagrangian and Frank Wolfe approximation. The algorithm was tested on on the Nguyen-Dupuis case with 13 nodes, 19 links.

Our contribution addresses critical limitations identified in prior studies by modeling EV users' responses via a bi-level framework, which uses UE as the LL problem. We develop an exact solution algorithm that achieves proven global optimality within practical runtimes on networks with up to 1,020 nodes and 2,512 links. We demonstrate the algorithm's efficiency through controlled benchmarks against two GAs on identical problem instances, where BPC consistently finds equal or better solutions in substantially less computation time. Since prior exact algorithms in the literature differ in formulation, objective structure, and network representation, direct head-to-head comparison based on reported network sizes or optimality gaps is not methodologically valid; our internal benchmarks therefore provide the primary basis for evaluating computational performance.

\section{Formulation}\label{Methodology}
We consider a bi-level formulation. The UL decision maker (network planner) designs an EV charging network and tries to maximize their revenue. At the same time, the LL represents the users, who choose charging stations to minimize their travel and charging costs, given the network design set by the planner. Because the planner’s revenue depends on the flow each station ultimately attracts, the UL objective is evaluated through the LL equilibrium. The planner's revenue depends on the users' choices, which will create a hierarchical optimization structure where the UL must solve the LL problem to evaluate any potential network design. For the route choice behavior of the LL problem, we assume that users follow UE principles, which means users seek the path with the minimum cost. Most papers whose model is traffic assignment make the same assumptions, too. The mathematical formulation of UE can be solved as a convex program because our model satisfies the necessary conditions: 1— We have assumed that users have perfect knowledge of the path cost; 2— The travel cost on a given link depends only on the flow on that link, and the travel cost functions are positive and monotonically increasing; 3— Each user selects the path with the least travel cost between their origin and destination.
Under these assumptions, the UE condition is characterized by the principle
that all used routes connecting the same origin-destination (OD) pair have equal and minimal travel cost. Unused routes may have higher travel
costs, and routes connecting different OD pairs may have different travel costs. A stable equilibrium is achieved when no user can reduce travel time by unilaterally changing routes. This formulation of the UE condition, first proposed by \citet{doi:10.1680/ipeds.1952.11259}, has been widely applied in transportation research. By solving the UE traffic assignment problem (TAP), we can output the route flows and each charging station flow that minimize travel cost in equilibrium across the network.

\subsection{LL Problem (User Charging Decision)}
In the LL problem, EV users choose their routes and charging stations in response to the infrastructure decisions made by the UL planner. Specifically, given the binary decision vector $\mathbf{y}$, which indicates whether a charging station is built at each candidate location. Users select feasible paths that minimize their combined travel and charging costs. The resulting traffic pattern must satisfy the UE condition: no user can improve their total cost by unilaterally switching to another available path, given the current congestion and charging infrastructure. This equilibrium reflects rational traveler behavior under perfect information and cost sensitivity. We assume travel time on each link is a function of flow, and is continuous and monotonically non-decreasing in flow, which ensures that the Beckmann objective is convex and the UE problem remains tractable. This assumption is standard in transportation modeling and is satisfied by commonly used functions such as the BPR function.

In real networks, not all EV travel demands between OD pairs can be satisfied due to limited battery range and sparse charging infrastructure. This is an issue for travel over long distances, since usually the trips within a city can be completed without needing to charge along the route. As a result, for trips with long distances, certain trips become infeasible if no reachable path exists that either stays within the battery limit or includes a charging opportunity. This constraint motivates defining the set of feasible paths based on battery limitations, which acknowledges that some OD pairs may have zero flow in the UE solution. The feasible paths available to each user depend on the installed charging infrastructure $\mathbf{y}$.

The LL problem thus jointly determines equilibrium path flows, link flows, and charging station demand, conditional on the UL decision $\mathbf{y}$. We consider a network $\mathcal{G} =(\mathcal{N}, \mathcal{A})$: graph $\mathcal{G}$ is the tuple $(\mathcal{N}, \mathcal{A})$ comprised of the set of nodes $\mathcal{N}$ and the set of links $\mathcal{A}$. Each link can be denoted by a pair of nodes $(i, j)$. Let  $\mathcal{Z} \subseteq \mathcal{N}$ be the set of zones, which represent nodes that are the origins or destinations of the travelers. $\mathcal{C} \subseteq \mathcal{N}$ is the set of candidate locations for building charging infrastructure. A path is feasible for an EV if either the usage of the battery is less than the battery limit, assuming that users will start their trip with a full battery, or there is a charging station on the route where the user recharges the battery to reach the destination before the battery is depleted. Therefore, not all demand from an origin to a destination can be met, and it depends on the $\mathbf{y}$, which specifies where charging stations are located and consequently which paths are feasible. Let $\Pi_{rs} (\mathbf{y})$ be the set of feasible paths based on the charging stations location $\mathbf{y}$ between the origin $r \in \mathcal{Z}$ and destination $s \in \mathcal{Z}$. Let $\Pi (\mathbf{y)} = \cup_{(r,s) \in \mathcal{Z}^2} \Pi_{rs} (\mathbf{y})$ be the set of all feasible paths. Travel time is considered the travel cost. The travel time for the link $(i, j)$ is based on a convex function of traffic flow on this link, which is $t_{ij} (x_{ij})$, where $x_{ij}$ is the link $(i, j)$ flow. Similarly, the waiting time due to the congestion and delay at station $l$ is modeled as a convex function $t_l(v_l)$ of the charging flow $v_l$. Before presenting the formulation, we summarize the key notation used throughout the model in Table \ref{tab:notation}.
\begin{table}[htbp]
\centering
\caption{Summary of Notation}
\label{tab:notation}
\begin{tabular}{cl}
\hline
\textbf{Notation} & \textbf{Description} \\
\hline
\multicolumn{2}{l}{\textit{Sets}} \\
$\mathcal{G} = (\mathcal{N}, \mathcal{A})$ & Transportation network with nodes $\mathcal{N}$ and links $\mathcal{A}$ \\
$\mathcal{Z} \subseteq \mathcal{N}$ & Set of zones (origins and destinations) \\
$\mathcal{C} \subseteq \mathcal{N}$ & Set of candidate locations for charging stations \\
$\Pi_{rs}(\mathbf{y})$ & Set of feasible paths from origin $r$ to destination $s$ under design $\mathbf{y}$ \\
$\Pi(\mathbf{y})$ & Set of all feasible paths under design $\mathbf{y}$ \\
\hline
\multicolumn{2}{l}{\textit{Parameters}} \\
$D_{rs}$ & Total potential demand between origin $r$ and destination $s$ \\
$d_{rs}(\mathbf{y})$ & Effective demand between $r$ and $s$ (depends on path feasibility) \\
$\delta_{ij}^{\pi}$ & Binary indicator: 1 if path $\pi$ uses link $(i,j)$, 0 otherwise \\
$\phi_{l}^{\pi}$ & Binary indicator: 1 if path $\pi$ includes charging at station $l$, 0 otherwise \\
$\gamma_{l}^{\pi}$ & Charging time at station $l$ for a vehicle on path $\pi$ \\
$c$ & Charging price per unit of time \\
$\tau$ & Value of time (converts charging duration to monetary cost) \\
$c^{\pi}$ & Total charging cost for path $\pi$: $c^{\pi} = \sum_{l \in \mathcal{C}} (c + \tau) \cdot \gamma_l^{\pi}$ \\
$t_{ij}(x_{ij})$ & Travel time function on link $(i,j)$ as a function of flow \\
$t_l(v_l)$ & Waiting time function at station $l$ as a function of charging flow \\
$t_{ij}^0$ & Free-flow travel time on link $(i,j)$ \\
$t_l^0$ & Base charging time at station $l$ (free-flow service time) \\
$C_{ij}$ & Capacity of link $(i,j)$ \\
$K_l$ & Service capacity of station $l$ \\
$C_l$ & Installation cost at candidate location $l$ \\
$B$ & Total available budget for charging station installation \\
$p$ & Charging revenue per unit of time \\
$w$ & Weight for unmet demand penalty in objective function \\
$M, M_l$ & Big-M constants \\
$\alpha, \beta$ & Shape parameters for BPR function \\
$\kappa$ & Capacity-cost proportionality constant: $K_l = \kappa \cdot C_l$ \\
$B^{\max}$ & Full battery capacity \\
$b_l^{\pi}$ & Arrival battery level at station $l$ for a vehicle on path $\pi$ \\
$\rho$ & Charging rate (time required to charge one unit of battery) \\
\hline
\multicolumn{2}{l}{\textit{Decision Variables}} \\
$y_l$ & Binary variable: 1 if charging station is built at location $l$, 0 otherwise \\
$h_{\pi}$ & Flow on path $\pi$ \\
$x_{ij}$ & Flow on link $(i,j)$ \\
$v_l$ & Total charging flow at station $l$ \\
\hline
\end{tabular}
\end{table}

The full formulation of this convex UE problem is given below.
\begin{subequations} \label{eq:ev_flow_optimization}
\begin{alignat}{2}
    \min_{\mathbf{x}, \mathbf{h}, \mathbf{v}} \quad 
    & L(\mathbf{y, x (h), v (h)}) = \sum_{(i, j) \in \mathcal{A}} \int_{0}^{x_{ij}} t_{ij}(x)dx 
    + \sum_{l \in \mathcal{C}} \int_{0}^{v_l} t_l(v) dv 
    + \sum_{\pi \in \Pi(\mathbf{y})} c^{\pi} h_{\pi}
    && \label{eq:ev_objective} \\
    \text{s.t.} \quad 
    & v_l = 0 \quad \text{if} \quad y_l = 0, 
    && \quad \forall l \in \mathcal{C} 
    \label{eq:ev_activation_limit} \\
    & x_{ij} = \sum_{\pi \in \Pi (\mathbf{y})} \delta_{ij}^{\pi} h_{\pi}, 
    && \quad \forall (i, j) \in \mathcal{A} 
    \label{eq:ev_link_flow_def} \\
    &  d_{rs}(\mathbf{y}) = \sum_{\pi \in \Pi_{rs} (\mathbf{y})} h_{\pi}, 
    && \quad \forall (r,s) \in \mathcal{Z}^2 
    \label{eq:ev_demand_satisfaction} \\
    & v_l = \sum_{\pi \in \Pi (\mathbf{y})} \phi_{l}^{\pi} h_{\pi}, 
    && \quad \forall l \in \mathcal{C} 
    \label{eq:ev_node_flow_def} \\
    & v_l \geq 0, 
    && \quad \forall l \in \mathcal{C} 
    \label{eq:ev_node_flow_nonneg} \\
    & h_{\pi} \geq 0.
    && \quad \forall \pi \in \Pi (\mathbf{y}) 
    \label{eq:ev_path_flow_nonneg}
\end{alignat}
\end{subequations}
The LL problem \eqref{eq:ev_flow_optimization} constitutes a modified traffic assignment problem (MTAP) that accounts for the unique cost structure faced by EV users. In contrast to classical TAP formulations, which only consider travel time as the cost, our model includes travel time, station congestion, and path-dependent charging costs, which indicates EV users must also select charging stations along their routes. This extension is reflected in the objective function, which includes three terms: link travel time, station congestion cost, and charging cost based on each path's energy requirements. Furthermore, two additional constraints distinguish this model from standard TAP: (i) a station activation constraint \eqref{eq:ev_activation_limit} ensures no users can charge at the location $l$ unless a charging facility is built there. In this constraint, $y_l$ comes from UL and is fixed here, and (ii) a flow conservation constraint that determines the amount of charging demand at each station based on users' path choices. These modifications are grounded in the assumption that EV travelers behave rationally under perfect information and choose paths that minimize their total travel and charging cost. 

We will name this problem MTAP ($\mathbf{y}$), and let $L(\mathbf{y},\mathbf{x}(\mathbf{h}), \mathbf{v}(\mathbf{h}))$ be the objective function value. 
To achieve the travel time between $r$ and $s$, we need to sum up all the links' travel time between $r$ and $s$, which are in the user's path. The link flows, and the path flows, denoted as $h_{\pi}$, are closely related. Let $\delta^{\pi}_{ij}$ denote whether link $(i, j)$ is used by path $\pi$, so $\delta^{\pi}_{ij}= 0$ if path $\pi$ does not use link $(i, j)$, and $\delta^{\pi}_{ij}= 1$ if it does. Based on constraint \eqref{eq:ev_link_flow_def}, link flows are the sum of all the path flows that use that link.

Constraint \eqref{eq:ev_demand_satisfaction} implies that each trip from origin $r$ to destination $s$ must take one of the available paths that exist in the network and is feasible considering the battery (i.e., the EV can traverse all links or recharge as needed). Therefore, the effective demand $d_{rs}(\mathbf{y})$ between origin $r$ and destination $s$ is defined as:

\begin{equation}
d_{rs}(\mathbf{y}) = \begin{cases} 
D_{rs} & \text{if at least one battery-feasible path exists from } r \text{ to } s \\
0 & \text{otherwise}
\end{cases}
\end{equation}
where $D_{rs}$ represents the total potential demand between origin $r$ and destination $s$. We assume that users only contribute to demand if there is at least one feasible path between their origin and destination. Consequently, we are making demand conditional on network feasibility rather than simply routing existing demand as in classical TAP formulations.

Let $v_l$ be the total flow of vehicles charging at station $l$ per unit time. We can only have a positive value for $v_l$ if the decision is to build a charging station there. Constraint \eqref{eq:ev_activation_limit} shows this linking relation between $v_l$ and $y_l$. So if $y_l$ is 0, then $v_l$ must be 0 too. In constraint \eqref{eq:ev_node_flow_def}, $\phi^{\pi}_{l} \in \{0,1\}$ indicates whether a vehicle traveling along path $\pi$ charges at station $l$, and $h_\pi$ is the flow on that path. Therefore, $v_l$ is equal to the total flow of vehicles that stop to charge at station $l$. To capture congestion effects at charging stations, we model the waiting and service time as a flow-dependent function $t_l(v_l)$, where higher demand leads to increased delays. This reflects the realistic scenario where stations with limited capacity experience congestion as more users arrive.

The third term in the objective function captures the path-dependent charging cost. Let $\gamma^{\pi}_{l}$ denote the amount of time (in hours) that a vehicle traveling along path $\pi$ must spend charging at station $l$, which depends on the battery level upon arrival. The total charging cost for path $\pi$ is defined as $c^{\pi} = \sum_{l \in \mathcal{C}} (c + \tau) \cdot \gamma_l^{\pi}$, where $c$ is the charging price per unit of time and $\tau$ is value of time which converts waiting time to recharge to equivalent cost. This formulation ensures that vehicles arriving with different battery states pay different amounts based on how much they actually need to charge.

The constraints \eqref{eq:ev_path_flow_nonneg} and \eqref{eq:ev_node_flow_nonneg} require the non-negativity for the flow of the charging station and the flow of the path.
The first term in the objective function of problem \eqref{eq:ev_flow_optimization} corresponds to the classical Beckmann transformation, which is a potential function whose minimizer satisfies the UE conditions with respect to travel time. The second term incorporates the station congestion cost through the integral $\int_0^{v_l} t_l(v) dv$, where $t_l(v)$ is a convex, monotonically increasing function representing the marginal waiting time at station $l$ under flow $v$. This formulation ensures that as station utilization increases, users experience longer waiting times, which is reflected in their path choices at equilibrium. The third term $\sum_{\pi} c^{\pi} h_{\pi}$ is linear in path flows and captures the actual energy cost paid by users and the time value they spend charging, which varies by path based on battery consumption before reaching the charging station.

%%%%%%%%%%%%%%%%%%%%%%%%%%%%%%%%%%%%%%%%%%%%%
\subsection{Bi-level Problem}
In every network, we have realistic limitations on feasible EV travel paths. In practice, not all EV travel demand can be satisfied due to inherent range restrictions and limited availability of charging infrastructure. Therefore, some OD pairs are infeasible when EVs lack sufficient battery charge or access to charging facilities along feasible routes. Consequently, these infeasible OD pairs naturally lead to a portion of unmet travel demand. The decision of whether to install a charging station at a specific location directly influences this amount of unmet demand. Locations without charging infrastructure can become bottlenecks. They will reduce accessibility and increase the unmet demand. Let $\Pi_{rs}$ be the set of all paths between the origin $r \in \mathcal{Z}$ and destination $s \in \mathcal{Z}$. Let $\Pi = \cup_{(r,s)} \Pi_{rs}$ be the set of all paths. To address this, we considered unmet demand into the UL optimization objective by adding a penalty term that represents the difference between the total demand $D_{rs}$ which also includes the infeasible demand and the actual served demand $ \sum_{\pi \in \Pi_{rs} (\mathbf{y})} h_{\pi}$. Specifically, the additional term $w \cdot \sum_{(r,s)} \left(D_{rs} - \sum_{\pi \in \Pi_{rs} (\mathbf{y})} h_{\pi}\right)$ incentivizes optimal charging station placement that minimizes unmet EV demand across the network, and $w$ is a weight to balance the cost of unmet demand with the value of charging. So the overall bi-level optimization model can now be written as:
\begin{subequations} \label{eq:station_placement_bi-level}
\begin{alignat}{2}
    \min_{\mathbf{y, x, v}} \quad 
    & -\sum_{l\in\mathcal{C}}p \cdot v_l + w \cdot \sum_{(r,s)} \left(D_{rs} - \sum_{\pi \in \Pi_{rs}} h_{\pi}\right)
    && \label{eq:station_obj_bi-level} \\[6pt]
    \text{s.t.} \quad 
    & \sum_{l \in \mathcal{C}} C_l \cdot y_l \leq B,
    && \label{eq:station_budget_bi-level} \\[6pt]
    & y_l \in \{0,1\},
    && \quad \forall l \in \mathcal{C}
    \label{eq:station_binary_bi-level} \\[6pt]
    & \mathbf{v, h} \in \text{MTAP} (\mathbf{y})
\end{alignat}
\end{subequations}
Let $y_l\in\{0,1\}$ be a binary variable equal to 1 if a charging facility is installed at location $l$, and 0 otherwise. Let $c_l$ be the installation cost at location $l$. Equation \eqref{eq:station_budget_bi-level} indicates that the total cost of building charging stations should be less than or equal to the available budget $B$. $p$ is the constant charging price per unit of time. The objective function \eqref{eq:station_obj_bi-level} indicates that the network designer would like to maximize the charging revenue and minimize the unmet demand.

It is important to note that the revenue term and the unmet demand penalty term capture fundamentally different aspects of the planning problem and do not represent double counting. The revenue term $\sum_{l\in\mathcal{C}}p \cdot v_l$ represents the actual charging revenue collected from EV users who successfully complete their trips and charge at the installed stations. In contrast, the unmet demand penalty captures the social cost of travelers who cannot complete their trips at all due to insufficient charging infrastructure. These two terms are mutually exclusive for any given demand: users either contribute to revenue (if they can travel and a feasible path exists) or to the unmet demand penalty (if no feasible path exists), but never both.

Furthermore, the specific objective function presented here represents one particular planning scenario. Our algorithmic framework is general and can accommodate alternative upper-level objectives depending on stakeholder interests. For instance, a private operator focused solely on profit could set $w = 0$, while a public agency prioritizing accessibility could increase $w$ or adopt entirely different objective functions such as minimizing total system travel time or maximizing network coverage.

\section{High Point Relaxation and Value Functions}\label{Solution Algorithm}
In this section, we develop relaxations of the bi-level EV charging station placement problem. We propose the HPR, which provides a valid global lower bound by replacing the LL UE objective with only its flow conservation and feasibility constraints. The resulting single level mixed integer formulation serves as a useful surrogate for bounding in the solution process. We then further build on this by introducing the value function cut to effectively cut the feasible solution values that are not optimal for LL.
\subsection{High Point Relaxation (Global Lower Bound)}
We will derive the HPR problem by considering only the follower's constraints without taking into account its objective function and relaxing the binary variable $y$ to be continuous.
\begin{subequations} \allowdisplaybreaks\label{eq:station_placement_bi-level2}
\begin{alignat}{2}
    \min_{\mathbf{y, x, v}} \quad 
    & -\sum_{l\in\mathcal{C}}p \cdot v_l + w \cdot \sum_{(r,s)} \left(D_{rs} - \sum_{\pi \in \Pi_{rs} (\mathbf{y})} h_{\pi}\right)
    && \label{eq:station_upper_objective} \\[6pt]
    \text{s.t.} \quad 
    & \sum_{l \in \mathcal{C}} C_l \cdot y_l \leq B,
    && \label{eq:station_budget_constraint} \\[6pt]
    & 0 \leq y_l\leq 1, 
    && \quad \forall l \in \mathcal{C}
    \label{eq:station_binary_constraint} \\[6pt]
    & v_l \leq M \cdot y_l, 
    && \quad \forall l \in \mathcal{C}
    \label{eq:bigM_constraint} \\[6pt]
    & x_{ij} = \sum_{\pi \in \Pi (\mathbf{y})} \delta_{ij}^{\pi} h_{\pi}, 
    && \quad \forall (i, j) \in \mathcal{A}
    \label{eq:link_flow_constraint} \\[6pt]
    & \sum_{\pi \in \Pi_{rs} (\mathbf{y})} h_{\pi} = d_{rs}(\mathbf{y}), 
    && \quad \forall (r,s) \in \mathcal{Z}^2
    \label{eq:demand_satisfaction_constraint} \\[6pt]
    & v_l = \sum_{\pi \in \Pi (\mathbf{y})} \phi_{l}^{\pi} h_{\pi}, 
    && \quad \forall l \in \mathcal{C}
    \label{eq:charging_flow_constraint} \\[6pt]
    & v_l \geq 0, 
    && \quad \forall l \in \mathcal{C}
    \label{eq:charging_nonnegative} \\[6pt]
    & h_{\pi} \geq 0, 
    && \quad \forall \pi \in \Pi
    \label{eq:path_flow_nonnegative}
\end{alignat}
\end{subequations}
\begin{theorem}\label{prop:HP_lower_bound}
Let $Z^{\star}$ be the optimal objective value of the bi-level EV–station-placement
problem \eqref{eq:station_placement_bi-level} and let
$Z^{\mathrm{HP}}$ be the optimal objective value of its
HPR \eqref{eq:station_placement_bi-level2}.
Then
\[
  Z^{\mathrm{HP}}\;\le\;Z^{\star}.
\]
\end{theorem}

\begin{proof}
Let $\mathcal F^{\text{BL}}$ and $\mathcal F^{\text{HP}}$ denote the feasible
sets of \eqref{eq:station_placement_bi-level} and
\eqref{eq:station_placement_bi-level2}, respectively.
The two problems share the same decision variables $(\mathbf y,\mathbf x, \mathbf{v})$ and enforce the same hard constraints.
The only element that distinguishes the bi-level model from the HPR is the follower optimality requirement: $\mathbf{v, x} \in \mathrm{MTAP}(\mathbf y)$ embedded in~\eqref{eq:station_placement_bi-level}.  Eliminating this requirement can only enlarge the set of feasible solutions.
Both problems minimize the same UL cost function
\eqref{eq:station_obj_bi-level} over their respective feasible sets. For any optimization problem, minimizing over a superset of feasible solutions cannot yield a larger objective value.  Hence $Z^{\mathrm{HP}} \leq Z^{\star}.$ 
\end{proof}

\subsection{Value Function Cuts}
If we solve the HPR, we may end up with a solution $(\mathbf{y}^1, \mathbf{v}^1, \mathbf{x}^1)$ where $(\mathbf{v}^1, \mathbf{x}^1) \notin \text{MTAP}(\mathbf{y}^1)$. This means the solution is not optimal for the LL.
Therefore, if we have $\mathbf{y}^1$ and solve the MTAP($\mathbf{y}$) problem, we can achieve a lower optimal value for the objective function with a different $(\mathbf{v}, \mathbf{x})$. This means that 
\begin{equation}\label{eq18}
    L(\mathbf{y}^1 , \mathbf{v}^1(\mathbf{h}^1), \mathbf{x}^1(\mathbf{h}^1)) > \min_{\mathbf{h}} L(\mathbf{y}^1, \mathbf{v}(\mathbf{h}), \mathbf{x}(\mathbf{h}))
\end{equation}
 Note that generally both $\mathbf{v, x}$ are derived based on one vector of $\mathbf{h}$, and in equation \eqref{eq18} $\mathbf{h}^1 \neq \mathbf{h}$. Based on $\mathbf{y}$, we have some feasible paths, not all of them. Let $\mathcal{H} (\mathbf{y})$ be the set of feasible path under $\mathbf{y}$.
$L(\mathbf{y}^1, \mathbf{v}^1(\mathbf{h}^1), \mathbf{x}^1(\mathbf{h}^1))$ is the value of the objective function at the LL by entering $(\mathbf{y}^1, \mathbf{v}^1(\mathbf{h}^1), \mathbf{x}^1(\mathbf{h}^1))$ and $L(\mathbf{y}^1, \mathbf{v}(\mathbf{h}), \mathbf{x}(\mathbf{h}))$ is the objective function of the optimal solution of the problem at the LL if we optimize it for the given $\mathbf{y}^1$.
So, to enforce the follower optimality, we will add the opposite of equation \eqref{eq18} to restrict $(\mathbf{v}, \mathbf{x})$. After adding this constraint to the problem \eqref{eq:station_placement_bi-level2}, we will have a problem which is the same as the original bi-level problem \eqref{eq:station_placement_bi-level}. Here is the HPR augmented with value function cuts:
\begin{subequations}\allowdisplaybreaks\label{eq:station_placement_bi-level3}
\begin{alignat}{2}
    \min_{\mathbf{y, x, v}} \quad 
    & -\sum_{l\in\mathcal{C}}p \cdot v_l + w \cdot \sum_{(r,s)} \left(D_{rs} - \sum_{\pi \in \Pi_{rs} (\mathbf{y})} h_{\pi}\right)
    && \label{eq:station_upper_objective3} \\[6pt]
    \text{s.t.} \quad 
    & \sum_{l \in \mathcal{C}} C_l \cdot y_l \leq B,
    && \label{eq:station_budget_constraint3} \\[6pt]
    & 0 \leq y_l \leq 1, 
    && \quad \forall l \in \mathcal{C}
    \label{eq:station_binary_constraint3} \\[6pt]
    & v_l \leq M \cdot y_l, 
    && \quad \forall l \in \mathcal{C}
    \label{eq:bigM_constraint3} \\[6pt]
    & x_{ij} = \sum_{\pi \in \Pi (\mathbf{y})} \delta_{ij}^{\pi} h_{\pi}, 
    && \quad \forall (i, j) \in \mathcal{A}
    \label{eq:link_flow_constraint3} \\[6pt]
    & \sum_{\pi \in \Pi_{rs} (\mathbf{y})} h_{\pi} = d_{rs}(\mathbf{y}), 
    && \quad \forall (r,s) \in \mathcal{Z}^2
    \label{eq:demand_satisfaction_constraint3} \\[6pt]
    & v_l = \sum_{\pi \in \Pi (\mathbf{y})} \phi_{l}^{\pi} h_{\pi}, 
    && \quad \forall l \in \mathcal{C}
    \label{eq:charging_flow_constraint3} \\[6pt]
    & v_l \geq 0, 
    && \quad \forall l \in \mathcal{C}
    \label{eq:charging_nonnegative3} \\[6pt]
    & h_{\pi} \geq 0, 
    && \quad \forall \pi \in \Pi
    \label{eq:path_flow_nonnegative3}\\[6pt]
    & L(\mathbf{y}, \mathbf{v}(\mathbf{h}), \mathbf{x}(\mathbf{h})) \leq L(\mathbf{y}, \mathbf{v}^{\prime} (\mathbf{h}^\prime), \mathbf{x}^{\prime} (\mathbf{h}^\prime)), 
&& \quad \forall \mathbf{h}^\prime \in \mathcal{H}(\mathbf{y})
\label{eq:lower_level_optimality_constraint3}
\end{alignat}
\end{subequations}
Constraint \eqref{eq:lower_level_optimality_constraint3} represents the value function cuts. $\mathcal{H}(\mathbf{y})$ represents the set of feasible path flow vectors $\mathbf{h}^\prime$ under the network design $\mathbf{y}$, and $\mathbf{v}^{\prime} (\mathbf{h}^\prime)$, $\mathbf{x}^{\prime} (\mathbf{h}^\prime)$ are the corresponding charging and link flows.
In particular, for each $ \mathbf{h}^\prime \in \mathcal{H}(\mathbf{y})$, the derived aggregate flows \(\mathbf{v}^\prime(\mathbf{h}^\prime)\) and \(\mathbf{x}^\prime(\mathbf{h}^\prime)\) respect \(v_l^\prime = 0\) if \(y_l = 0\) for all \(l \in \mathcal{C}\), so that no site without an established charging station is assigned any charging demand. By restricting to flow configurations consistent with the given \(\mathbf{y}\), this formulation avoids infeasible cases such as \(v_l^\prime > 0\) when \(y_l = 0\).
\begin{theorem}\label{proposition2}
By adding value function cut \eqref{eq:lower_level_optimality_constraint3} to the HPR problem and restoring the integrality constraint on $\mathbf{y}$, problem \eqref{eq:station_placement_bi-level3} is equivalent to problem \eqref{eq:station_placement_bi-level}.
\end{theorem}
\begin{proof}
It is sufficient to show equality of the feasible regions, since the objective functions in both problems are identical. Suppose $(\mathbf{y}, \mathbf{v}, \mathbf{x})$ is feasible for problem \eqref{eq:station_placement_bi-level3}. This means $(\mathbf{y}, \mathbf{v}, \mathbf{x})$ satisfies the value function cuts in constraint \eqref{eq:lower_level_optimality_constraint3}. Therefore, under the same $\mathbf{y}$, the $L(\mathbf{y}, \mathbf{v}(\mathbf{h}), \mathbf{x}(\mathbf{h}))$ is less than or equal to the $L(\mathbf{y}, \mathbf{v}^{\prime} (\mathbf{h}^\prime), \mathbf{x}^{\prime}(\mathbf{h}^\prime))$  for all $ (\mathbf{v}^\prime(\mathbf{h}^\prime) , \mathbf{x}^\prime (\mathbf{h}^\prime))$ that are derived form $\mathbf{h}^\prime \in \mathcal{H}(\mathbf{y})$. This directly means $(\mathbf{v}, \mathbf{x})$ minimizes $L(\mathbf{y}, \mathbf{v}(\mathbf{h}), \mathbf{x}(\mathbf{h}))$ given the decision $\mathbf{y}$. Thus, we can ensure the solution is feasible for the original bi-level problem \eqref{eq:station_placement_bi-level3}. Conversely, consider any point $(\mathbf{y}, \mathbf{v}, \mathbf{x})$ feasible for the original bi-level problem \eqref{eq:station_placement_bi-level3}. By the definition of LL optimality, $(\mathbf{v}, \mathbf{x})$ minimizes $L(\mathbf{y}, \mathbf{v}(\mathbf{h}), \mathbf{x}(\mathbf{h}))$ given the decision $\mathbf{y}$. Therefore, constraint \eqref{eq:lower_level_optimality_constraint3} is trivially satisfied, implying feasibility for problem \eqref{eq:station_placement_bi-level3}. Hence, both problems have identical feasible sets and are thus equivalent. 
\end{proof}

Constraint \eqref{eq:lower_level_optimality_constraint3}  requires evaluating all possible flow configurations \((\mathbf{v}^\prime, \mathbf{x}^\prime\)), which is computationally intractable. To overcome this, we adopt a cutting plane strategy that iteratively adds only a subset of these constraints. Specifically, in each iteration, we generate a valid cut based on a carefully chosen pair \((\mathbf{v}^\mathrm{f}, \mathbf{x}^\mathrm{f})\). To construct this subset, we leverage the solution of the MTAP$(\mathbf{y})$ problem, which represents the optimal response of the LL equilibrium model for a given network design \(\mathbf{y}\). Since MTAP$(\mathbf{y})$ provides the minimal value of the LL objective function under \(\mathbf{y}\), using its solution to define \((\mathbf{v}^\mathrm{f}, \mathbf{x}^\mathrm{f})\) ensures that the resulting inequality is an effective value function cut. When the HPR gives us a solution $(\mathbf{y}^1, \mathbf{v}^1, \mathbf{x}^1)$ where $(\mathbf{v}, \mathbf{x}) \notin \arg \min_{\mathbf{v}, \mathbf{x} }L(\mathbf{y}^1, \mathbf{v}(\mathbf{h}), \mathbf{x}(\mathbf{h}))$, let $(\mathbf{v}^\mathrm{f}, \mathbf{x}^\mathrm{f}) \in \text{MTAP}(\mathbf{y}^1)$, then if we add the following cut
\begin{equation}\label{eq7}
    L(\mathbf{y}, \mathbf{v}(\mathbf{h}), \mathbf{x}(\mathbf{h})) \leq L(\mathbf{y}, \mathbf{v}^\mathrm{f}(\mathbf{h}^\mathrm{f}), \mathbf{x}^\mathrm{f}(\mathbf{h}^\mathrm{f})) \quad \forall (\mathbf{v}^\mathrm{f}(\mathbf{h}^\mathrm{f}) , \mathbf{x}^\mathrm{f} (\mathbf{h}^\mathrm{f})) \in \text{MTAP}(\mathbf{y})
\end{equation}
to the HPR, it will exclude answers like $(\mathbf{y}^1, \mathbf{v}^1, \mathbf{x}^1)$.
If we add cuts sequentially, this means that we will exclude the $\mathbf{v} \in \mathbf{V}$ and $\mathbf{x} \in \mathbf{X}$, which are feasible for the HPR problem, but not optimal or feasible for the LL problem. By adding these cuts, we will tighten the feasible region of the HPR problem, therefore achieving a lower gap between the lower bound and the original problem solution. This means that we will have a set of $(\mathbf{v}^\mathrm{f}(i), \mathbf{x}^\mathrm{f}(i)) $ which are built sequentially over the iterations. 

While value function cuts are useful for iteratively refining the feasible region of the HPR, their applicability is conditioned on the feasibility of the follower solution under the current UL decision \(\mathbf{y}\). In particular, a previously generated pair \((\mathbf{v}^\mathrm{f}, \mathbf{x}^\mathrm{f}) \in \text{MTAP}(\mathbf{y}^\prime)\) may no longer be valid if certain charging stations used in that solution are deactivated in the new \(\mathbf{y}\). To prevent invalid comparisons and over constraining the feasible region, we introduce a relaxation term using big $M$ constants. These terms deactivate the cut whenever a required charging station is not selected, thus ensuring that the value function cut only applies when the previous follower solution remains feasible. Let $M_l$ be a sufficiently large constant for each candidate location $l \in \mathcal{C}$, and define the set $\mathcal{C}^\mathrm{f} := \{ l \in \mathcal{C} : v_l^\mathrm{f} > 0 \}$, i.e., the subset of candidate locations that were used in the previous follower solution. Let $\mathcal{F}$ denote the index set of previously generated follower solutions, where each $i \in \mathcal{F}$ corresponds to a solution $(\mathbf{v}^\mathrm{f}(i), \mathbf{x}^\mathrm{f}(i), \mathbf{h}^\mathrm{f}(i))$ obtained from solving MTAP$(\mathbf{y})$ at some iteration of the algorithm. The relaxed value function cuts become:
\begin{equation}\label{cons9}
    L(\mathbf{y}, \mathbf{v}(\mathbf{h}), \mathbf{x}(\mathbf{h})) \leq L(\mathbf{y}, \mathbf{v}^\mathrm{f}(i), \mathbf{x}^\mathrm{f}(i)) + \sum_{l \in \mathcal{C}^\mathrm{f}} M_l (1 - y_l), \quad \forall i \in \mathcal{F}
\end{equation}
This form ensures that if any $l \in \mathcal{C}$ is closed in the current $\mathbf{y}$ (i.e., $y_l = 0$), then the corresponding term $M_l(1 - y_l)$ becomes active and relaxes the constraint, which will effectively deactivate the cut. On the other hand, when all previously used stations are still open ($y_l = 1$ for all $l \in \mathcal{C}$), the big-$M$ terms vanish, and the cut enforces that the current LL solution is not strictly better than the previously observed one. 
\begin{subequations}\allowdisplaybreaks\label{eq:station_placement_bi-level4}
\begin{alignat}{2}
    \min_{\mathbf{y, x, v}} \quad 
    & -\sum_{l\in\mathcal{C}}p \cdot v_l + w \cdot \sum_{(r,s)} \left(D_{rs} - \sum_{\pi \in \Pi_{rs} (\mathbf{y})} h_{\pi}\right)
    && \label{eq:station_upper_objective4} \\[6pt]
    \text{s.t.} \quad 
    & \sum_{l \in \mathcal{C}} C_l \cdot y_l \leq B,
    && \label{eq:station_budget_constraint4} \\[6pt]
    & 0 \leq y_l \leq 1, 
    && \quad \forall l \in \mathcal{C}
    \label{eq:station_binary_constraint4} \\[6pt]
    & v_l \leq M \cdot y_l, 
    && \quad \forall l \in \mathcal{C}
    \label{eq:bigM_constraint4} \\[6pt]
    & x_{ij} = \sum_{\pi \in \Pi (\mathbf{y})} \delta_{ij}^{\pi} h_{\pi}, 
    && \quad \forall (i, j) \in \mathcal{A}
    \label{eq:link_flow_constraint4} \\[6pt]
    & \sum_{\pi \in \Pi_{rs} (\mathbf{y})} h_{\pi} = d_{rs}(\mathbf{y}), 
    && \quad \forall (r,s) \in \mathcal{Z}^2
    \label{eq:demand_satisfaction_constraint4} \\[6pt]
    & v_l = \sum_{\pi \in \Pi (\mathbf{y})} \phi_{l}^{\pi} h_{\pi}, 
    && \quad \forall l \in \mathcal{C}
    \label{eq:charging_flow_constraint4} \\[6pt]
    & v_l \geq 0, 
    && \quad \forall l \in \mathcal{C}
    \label{eq:charging_nonnegative4} \\[6pt]
    & h_{\pi} \geq 0, 
    && \quad \forall \pi \in \Pi
    \label{eq:path_flow_nonnegative4}\\[6pt]
    & L(\mathbf{y}, \mathbf{v}(\mathbf{h}), \mathbf{x}(\mathbf{h})) \leq L(\mathbf{y}, \mathbf{v}^{\mathrm{f}} (i), \mathbf{x}^{\mathrm{f}} (i))+ \sum_{l \in \mathcal{C}^\mathrm{f}} M_l (1 - y_l)
    && \quad \forall i \in \mathcal{F}.
    \label{eq:lower_level_optimality_constraint4}
\end{alignat}
\end{subequations}
\begin{theorem}
After adding cut \eqref{eq:lower_level_optimality_constraint4} with a specific $(\mathbf{v}^\mathrm{f}, \mathbf{x}^\mathrm{f})$, the optimal solution of problem \eqref{eq:station_placement_bi-level4} is still a lower bound on the problem \eqref{eq:station_placement_bi-level}.
\end{theorem}
\begin{proof}
    In constraint \eqref{eq:lower_level_optimality_constraint4}, we are just considering a subset of $(\mathbf{v}, \mathbf{x})$ compared to the constraint \eqref{eq:lower_level_optimality_constraint3} in problem \eqref{eq:station_placement_bi-level3}, which based on Proposition \ref{proposition2} is equivalent to the original bi-level problem \eqref{eq:station_placement_bi-level}. So the optimal solution of this problem is still a lower bound for problem \eqref{eq:station_placement_bi-level}.
\end{proof}

\subsubsection{Outer Approximation for Value Function Cuts}
To remove the nonlinearity in constraint~\eqref{eq:lower_level_optimality_constraint4}, we apply the OA method. This method was originally proposed by \citet{Duran1986}. We can use OA when the constraint involves a convex function. The follower’s objective function \(L(\mathbf{y}, \mathbf{v}(\mathbf{h}), \mathbf{x}(\mathbf{h}))\), derived from a convex UE traffic assignment problem, is convex in the continuous variables \(\mathbf{v}\) and \(\mathbf{x}\). The OA comes from linearizing the convex functions at previously found solutions, which generates supporting hyperplanes. These hyperplanes will iteratively refine the objective function \cite{Duran1986}. Therefore, at a given point \((\hat{\mathbf{v}}, \hat{\mathbf{x}})\), we can linearize the follower’s objective by constructing the first order Taylor expansion.
By iteratively adding such cuts from previously computed follower solutions, we progressively refine a linear OA of the nonlinear constraint. As demonstrated in \citet{Duran1986}, this iterative process ensures that the collection of linearizations converges to fully approximate the original nonlinear feasible region. Although many cuts may be needed to tightly approximate the feasible region, this method provides a tractable way to enforce follower optimality without solving the full nonlinear bi-level problem. This OA technique enables us to convert the bi-level model into a sequence of MILPs that are significantly more computationally efficient to solve.
The OA of the left-hand side of constraint \eqref{eq:lower_level_optimality_constraint4} is
\begin{equation}
    OA[L(\mathbf{y}, \mathbf{v}(\mathbf{h}), \mathbf{x}(\mathbf{h}))] \leq L(\mathbf{y}, \mathbf{v}(\mathbf{h}), \mathbf{x}(\mathbf{h}))
\end{equation}
We will derive the OA around $(\mathbf{y}^1, \mathbf{v}^1, \mathbf{x}^1)$. Note that $L$ is composed of three parts: the first two parts are nonlinear, while the third part (charging cost) is linear. We will derive the OA for both nonlinear parts since the OA of a linear term is equal to itself. 
For the first term, if we consider the left-hand side of constraint \eqref{eq:lower_level_optimality_constraint4} and derive the OA:
\begin{equation}\label{eq:long_constraint}
OA\left(\sum_{(i, j) \in \mathcal{A}} \int_{0}^{x_{ij}} t_{ij}(x)dx \right)
 \leq \sum_{(i, j) \in \mathcal{A}} \int_{0}^{x_{ij}} t_{ij}(x)dx
\end{equation}
The OA of a convex function $f(x)$ at a point $x^1$ is given by the first-order Taylor expansion: $f(x^1) + \nabla f(x^1)^\top (x - x^1) \leq f(x)$. Applying this to the link travel time integral, where the derivative of $\int_{0}^{x_{ij}} t_{ij}(x)dx$ with respect to $x_{ij}$ is $t_{ij}(x_{ij})$, we obtain:
\begin{equation}\label{eq2011}
    \sum_{(i, j) \in \mathcal{A}} \int_{0}^{x_{ij}^1} t_{ij}(x)dx  
    + \sum_{(i, j) \in \mathcal{A}} t_{ij}(x_{ij}^1) \cdot (x_{ij} - x_{ij}^1) 
     \leq \sum_{(i, j) \in \mathcal{A}} \int_{0}^{x_{ij}} t_{ij}(x)dx
\end{equation}
Equation \eqref{eq2011} is separable by link $(i, j)$, meaning the inequality can be decomposed into independent terms for each link. To facilitate this decomposition, we introduce auxiliary variables $\eta_{ij}$ for each link $(i,j) \in \mathcal{A}$, which serve as upper bounds on the true integral $\int_{0}^{x_{ij}} t_{ij}(x)dx$. By requiring that the OA at any linearization point $x_{ij}^1$ is bounded by $\eta_{ij}$, we can write the link-specific OA constraints as:
\begin{equation}\label{eq23}
    \int_{0}^{x_{ij}^1} t_{ij}(x)dx + t_{ij}(x_{ij}^1) \cdot (x_{ij} - x_{ij}^1)  \leq \eta_{ij} \quad \forall x_{ij}^1
\end{equation}
These constraints ensure that $\eta_{ij} \geq \int_{0}^{x_{ij}} t_{ij}(x)dx$ for all $(i,j) \in \mathcal{A}$, and consequently:
\begin{equation}\label{eq:eta_sum}
    \sum_{(i, j) \in \mathcal{A}} \eta_{ij} \geq OA\left(\sum_{(i, j) \in \mathcal{A}} \int_{0}^{x_{ij}} t_{ij}(x)dx \right)
\end{equation}
Similarly, for the station congestion term, we derive the OA:
\begin{equation}\label{eq:station_oa_ineq}
OA\left(\sum_{l \in \mathcal{C}} \int_{0}^{v_l} t_l(v)dv \right)
 \leq \sum_{l \in \mathcal{C}} \int_{0}^{v_l} t_l(v)dv
\end{equation}
Applying the first-order Taylor expansion, where the derivative of $\int_{0}^{v_l} t_l(v)dv$ with respect to $v_l$ is $t_l(v_l)$:
\begin{equation}\label{eq:station_oa_expanded1}
    \sum_{l \in \mathcal{C}} \int_{0}^{v_l^1} t_l(v)dv  
    + \sum_{l \in \mathcal{C}} t_l(v_l^1) \cdot (v_l - v_l^1) 
     \leq \sum_{l \in \mathcal{C}} \int_{0}^{v_l} t_l(v)dv
\end{equation}
Equation \eqref{eq:station_oa_expanded1} is separable by station $l$. We introduce auxiliary variables $\psi_l$ for each station $l \in \mathcal{C}$, which serve as upper bounds on the true integral $\int_{0}^{v_l} t_l(v)dv$. By requiring that the OA at any linearization point $v_l^1$ is bounded by $\psi_l$, we can write the station-specific OA constraints as:
\begin{equation}\label{eq:station_oa_cut}
    \int_{0}^{v_l^1} t_l(v)dv + t_l(v_l^1) \cdot (v_l - v_l^1)  \leq \psi_l \quad \forall v_l^1
\end{equation}
These constraints ensure that $\psi_l \geq \int_{0}^{v_l} t_l(v)dv$ for all $l \in \mathcal{C}$, and consequently:
\begin{equation}\label{eq:psi_sum}
    \sum_{l \in \mathcal{C}} \psi_l \geq OA\left(\sum_{l \in \mathcal{C}} \int_{0}^{v_l} t_l(v)dv \right)
\end{equation}
We now show how to combine these results to linearize constraint \eqref{eq:lower_level_optimality_constraint4}. Recall that the value function cut requires:
\begin{equation}\label{eq:vf_recall}
    L(\mathbf{y}, \mathbf{v}(\mathbf{h}), \mathbf{x}(\mathbf{h})) \leq L(\mathbf{y}, \mathbf{v}^\mathrm{f}, \mathbf{x}^\mathrm{f}) + \sum_{l \in \mathcal{C}^\mathrm{f}} M_l (1 - y_l)
\end{equation}
Both sides of constraint \eqref{eq:lower_level_optimality_constraint4} are nonlinear, but the right-hand side is a constant since it is evaluated at the UE optimal solution $(\mathbf{x}^\mathrm{f}, \mathbf{v}^\mathrm{f})$. Therefore, we derive the OA only for the left-hand side, and only for the nonlinear components (the charging cost term $\sum_{\pi \in \Pi(\mathbf{y})} c^{\pi} h_{\pi}$ is already linear).
Expanding the left-hand side using the definition of $L(\cdot)$ and applying OA to the nonlinear terms:
\begin{multline}\label{eq:long_constraint2}
OA\left( \sum_{(i, j) \in \mathcal{A}} \int_{0}^{x_{ij}} t_{ij}(x)dx \right)
    + OA\left( \sum_{l \in \mathcal{C}} \int_{0}^{v_l} t_l(v) dv \right) + \sum_{\pi \in \Pi (\mathbf{y})} c^{\pi} h_{\pi} \\[6pt]
\leq \sum_{(i, j) \in \mathcal{A}} \int_{0}^{x_{ij}^\mathrm{f}} t_{ij}(x)dx 
+ \sum_{l \in \mathcal{C}} \int_{0}^{v_l^\mathrm{f}} t_l(v) dv + \sum_{\pi \in \Pi (\mathbf{y})} c^{\pi} h_{\pi}^\mathrm{f} + \sum_{l \in \mathcal{C}^\mathrm{f}} M_l (1 - y_l)
\end{multline}
From inequalities \eqref{eq:eta_sum} and \eqref{eq:psi_sum}, we have established that the auxiliary variables provide upper bounds on the OA terms:
\begin{equation}\label{eq:oa_bounds}
    OA\left( \sum_{(i, j) \in \mathcal{A}} \int_{0}^{x_{ij}} t_{ij}(x)dx \right) \leq \sum_{(i, j) \in \mathcal{A}} \eta_{ij}, \quad OA\left( \sum_{l \in \mathcal{C}} \int_{0}^{v_l} t_l(v) dv \right) \leq \sum_{l \in \mathcal{C}} \psi_l
\end{equation}
Substituting these bounds into \eqref{eq:long_constraint2}, we obtain the linearized value function cut:
\begin{equation}\label{eq:VFcut_bigM}
\sum_{(i, j) \in \mathcal{A}} \eta_{ij} + \sum_{l \in \mathcal{C}} \psi_l + \sum_{\pi \in \Pi (\mathbf{y})} c^{\pi} h_{\pi}
\leq  \sum_{(i, j) \in \mathcal{A}} \int_{0}^{x_{ij}^\mathrm{f}} t_{ij}(x)dx + \sum_{l \in \mathcal{C}} \int_{0}^{v_l^\mathrm{f}} t_l(v) dv + \sum_{\pi \in \Pi (\mathbf{y})} c^{\pi} h_{\pi}^\mathrm{f} + \sum_{l \in \mathcal{C}^\mathrm{f}} M_l (1 - y_l)
\end{equation}
Therefore, we have approximated cut \eqref{eq:lower_level_optimality_constraint4} in problem \eqref{eq:station_placement_bi-level4} with the combination of cuts \eqref{eq23}, \eqref{eq:station_oa_cut}, and \eqref{eq:VFcut_bigM}. The auxiliary variables $\eta_{ij}$ and $\psi_l$ are constrained by the OA cuts \eqref{eq23} and \eqref{eq:station_oa_cut}, which are iteratively added at each linearization point encountered during the algorithm. This approach transforms the nonlinear value function constraint into a collection of linear constraints that can be efficiently handled by MILP solvers.

The final linearized HPR will be the following:
\begin{subequations}\allowdisplaybreaks\label{eq:station_placement_bi-level5}
\begin{alignat}{2}
    \min_{\mathbf{y, x, v}, \pmb{\eta}, \pmb{\psi}} \quad 
    & -\sum_{l \in \mathcal{C}} p \cdot v_l + w \cdot \sum_{(r,s)} \left(D_{rs} - \sum_{\pi \in \Pi_{rs} (\mathbf{y})} h_{\pi}\right)
    && \label{eq:station_upper_objective5} \\[6pt]
    \text{s.t.} \quad 
    & \sum_{l \in \mathcal{C}} C_l \cdot y_l \leq B,
    && \label{eq:station_budget_constraint5} \\[6pt]
    & 0 \leq y_l \leq 1, 
    && \quad \forall l \in \mathcal{C}
    \label{eq:station_binary_constraint5} \\[6pt]
    & v_l \leq M \cdot y_l, 
    && \quad \forall l \in \mathcal{C}
    \label{eq:bigM_constraint5} \\[6pt]
    & x_{ij} = \sum_{\pi \in \Pi (\mathbf{y})} \delta_{ij}^{\pi} h_{\pi}, 
    && \quad \forall (i, j) \in \mathcal{A}
    \label{eq:link_flow_constraint5} \\[6pt]
    & \sum_{\pi \in \Pi_{rs} (\mathbf{y})} h_{\pi} = d_{rs}(\mathbf{y}), 
    && \quad \forall (r,s) \in \mathcal{Z}^2
    \label{eq:demand_satisfaction_constraint5} \\[6pt]
    & v_l = \sum_{\pi \in \Pi (\mathbf{y})} \phi_l^{\pi} h_{\pi}, 
    && \quad \forall l \in \mathcal{C}
    \label{eq:charging_flow_constraint5} \\[6pt]
    & v_l \geq 0, 
    && \quad \forall l \in \mathcal{C}
    \label{eq:charging_nonnegative5} \\[6pt]
    & h_{\pi} \geq 0, 
    && \quad \forall \pi \in \Pi
    \label{eq:path_flow_nonnegative5} \\[6pt]
    & \int_{0}^{x_{ij}^1} t_{ij}(x)dx + t_{ij} (x^1)  \cdot (x_{ij} - x_{ij}^1)  \leq \eta_{ij}, 
    && \quad \forall \mathbf{x}^1
    \label{eq:link_oa_constraint5} \\[6pt]
    & \int_{0}^{v_l^1} t_l(v)dv + t_l(v^1) \cdot (v_l - v_l^1)  \leq \psi_l, 
    && \quad \forall \mathbf{v}^1
    \label{eq:station_oa_constraint5} \\[6pt]
    & \sum_{(i,j)\in \mathcal{A}} \eta_{ij} + \sum_{l \in \mathcal{C}} \psi_l + \sum_{\pi \in \Pi (\mathbf{y})} c^{\pi} h_{\pi} \notag\\
    &\leq  \sum_{(i,j)\in \mathcal{A}} \int_{0}^{x_{ij}^\mathrm{f}} t_{ij}(x)dx 
    + \sum_{l \in \mathcal{C}} \int_{0}^{v_l^\mathrm{f}} t_l(v) dv 
    + \sum_{\pi \in \Pi (\mathbf{y})} c^{\pi} h_{\pi}^\mathrm{f}
    + \sum_{l \in \mathcal{C}^\mathrm{f}} M_l (1 - y_l)
    && \quad \forall (\mathbf{x}^\mathrm{f}, \mathbf{v}^\mathrm{f}, \mathbf{h}^\mathrm{f}) \in \mathcal{F}\label{vfcut5}
\end{alignat}
\end{subequations}

\section{Branch-and-Price-and-Cut Algorithm}\label{BPC}
We introduce a BPC algorithm to solve our problem. We have a single tree in which we obtain the LB by solving problem \eqref{eq:station_placement_bi-level5}. To obtain a feasible solution (and thus an upper bound) for any fixed leader decision \(\mathbf{y}\), we solve MTAP\((\mathbf{y})\). This problem corresponds to the LL UE problem, where users choose their routes and charging stations based on the charging infrastructure \(\mathbf{y}\) and the resulting travel costs. The solution to MTAP\((\mathbf{y})\) yields a set of path flows \(\mathbf{h}^\mathrm{f}\), link flows \(\mathbf{x}^\mathrm{f}\), and station demands \(\mathbf{v}^\mathrm{f}\) that are optimal from the users' perspective and feasible under the given \(\mathbf{y}\). We then evaluate the UL objective using this solution, which gives a feasible objective value corresponding to the bi-level problem. Since this value respects all constraints and represents a valid leader and follower interaction, it provides a valid UB in our BPC framework.

Let $k$ denote a node in the BPC and its branches. We will use the labeling scheme to keep track of the candidate location to place the charging station in a set of $C_2$. We will define two sets, one for locations where the decision is not to place a charging station $(y_l = 0)$ and another for locations where we decided to place a charging station ($y_l = 1$). These sets are defined for each node in the BPC tree.
\begin{subequations}
\begin{align}
    C_2^0 (k) = \{ l \in C_2: y_l ^k = 0 \}\\
    C_2^1 (k) = \{ l \in C_2 : y_l ^k = 1 \}
    \end{align}    
\end{subequations}
The locations where we want to decide whether to build a charging station or not are a subset of all of the locations:
\begin{equation}
    (C_2^0 (k)\cup C_2^1 (k) ) \subseteq C_2
\end{equation}
At the root node of the tree, both sets are empty. We will perform a check process for each node $k$ of the tree as shown below
\begin{itemize}
    \item if $\sum_{l \in C_2^1(k)} c_l . y_l(k) > B$, the node of BPC is infeasible since the budget constraint is not satisfied.
    \item if $|C_2 ^ 0 (k) \cup C_2 ^1 (k)| = |C_2|$, all the locations are fixed.
    \item if $B - \sum_{l \in C_2^1(k)} c_l . y_l(k) < 
\min \left\{ c_l : l \in C_2 \setminus \left( C_2^0(k) \cup C_2^1(k) \right) \right\}$, then there is not enough budget to build a new charging station at unfixed location $l$, so for all of the remaining locations $y_{l}(k) = 0$, and the BPC node is fixed.
    \item else the BPC node is unfixed.
\end{itemize}
If the result of this checking process is an unfixed node of BPC, we will solve the LB using problem \eqref{eq:station_placement_bi-level5}. If the result is fixed, we will obtain a UB by solving the LL problem  MTAP ($\mathbf{y}$). Algorithm \ref{BPCAlgorithm} shows the solution algorithm, and it mentions CG, which will be discussed in the next section. $\mathcal{B}$ is the set of nodes of BPC tree. 
\subsection{Solution Algorithm}
\begin{algorithm}[H]
\caption{BPC Algorithm}
\label{BPCAlgorithm}
\DontPrintSemicolon

Initialize \( UB \gets \infty \), \( LB \gets - \infty \), \( \mathcal{B} \gets \emptyset \)\;
%Initialize B\&B tree with root node 0\\
$LB_0$, $\mathbf{y}^0$, $\mathbf{v}^0 \gets$ Solve HPR~\eqref{eq:station_placement_bi-level5} via CG Algorithm \ref{alg:column_generation}\;
Add root node $k = 0 $ of BPC tree to $\mathcal{B}$\\
\While{\( \mathcal{B} \neq \emptyset \)}{
    remove node \( k \) with minimum LB and \emph{check($k$)}\;
    
    {}{
        \If{check($k$) = fixed}{
            \( UB_k \gets \text{Solve MTAP}(\mathbf{y}^k) \)\;
            \( UB \gets \min(UB_k, UB) \)\;
            Add VF cuts \eqref{eq:station_oa_constraint5}, \eqref{eq:link_oa_constraint5}, \eqref{vfcut5}\;
            %Add interdiction cut \eqref{interdictioncut}\;
            \If{\( LB_k \ge UB \)}{
                continue\;
            }
        }
        \ElseIf{check(\( k \)) = unfixed}{
            $LB_k$, $\mathbf{y}^k$, $\mathbf{v}^k \gets$ Solve HPR~\eqref{eq:station_placement_bi-level5} at node \( k \) via CG Algorithm \ref{alg:column_generation}\;
            
            %Let \( LB_k, \mathbf{v}^k, \mathbf{y}^k \) be the result\;
            \If{HPR infeasible or \( LB_k \ge UB \)}{
                continue\;
            }
            \Else{
                %\( LB \gets \min\{\text{LB}_k : k{\text{ is active}} \} \)\;
                \( C_2^{\text{frac}}(k) \gets \{ l \in C_2 : y_l^k \text{ fractional} \} \)\;
                Branch on a fractional variable \( y_l^k \in C_2^{\text{frac}}(k) \) with highest $v_l$: add two child nodes to \( \mathcal{B} \) with additional constraints \( y_l^{k+1} = 1 \) and \( y_l^{k+1} = 0 \), respectively.

            }
        }
        \ElseIf{check(\( k \))= infeasible}{
            continue\;
        }
        Update $LB$ based on active BB nodes: \( LB \gets \min\{\text{LB}_k : k{\text{ is active}} \} \)\ \\
        \If{\( UB - LB \le \epsilon \)}{
            Stop\;
        }
    }
}
\end{algorithm}
\begin{theorem}
    The BPC algorithm \ref{BPCAlgorithm} solves the DNDP to optimality in a finite number of iterations.
\end{theorem}
\begin{proof}
If \texttt{check}$(k)$ classifies a node $k$ as \emph{fixed}, all binary variables $y_l$ are determined. We then solve the follower equilibrium
\(\text{MTAP}(\mathbf y^{k})\) and obtain the feasible objective value
\(
  U\!B_{k}.
\)
Keeping the best such value over all nodes already evaluated gives
\[
  U\!B=\min\bigl\{U\!B_{k}\;:\;\text{$k$ is an active node of BPC tree}\bigr\},
\]
which is always a valid upper bound on the optimum of DNDP.
For an \emph{unfixed} node, the algorithm selects a free index
\(l^\star\) whose variable \(y_{l^\star}\) is fractional and whose
station demand \(v_{l^\star}\) is largest.
Two child nodes are created with the additional constraints
\(y_{l^\star}=1\) and \(y_{l^\star}=0\), respectively.
So no feasible binary vector is lost.
A node is discarded in exactly three cases:
\begin{enumerate}[label=(\alph*)]
  \item Infeasible.
  \item Bounded.
        If \(L\!B_{k}\ge U\!B\).
  \item Fixed. 
        Solving MTAP yields the exact value \(U\!B_{k}\); if
        \(U\!B_{k}\ge U\!B\) the subtree cannot contain a better solution.
\end{enumerate}
In each case, fathoming is safe; the global optimum cannot lie in a discarded subtree.
There are at most \(2^{n}\) complete assignments of the binary vector \(\mathbf y\), so the BPC tree is finite.
Each node is processed at most once. Therefore, the algorithm stops after a finite number of iterations.
When the algorithm terminates we have either an empty candidate list (all nodes fathomed), or \(U\!B-L\!B\le\epsilon\).
In first case, we have \(U\!B=L\!B\) = optimal objective value.  
Otherwise\[
  L\!B\;\le\;\text{optimal objective value}\;\le\;U\!B\;\le\;L\!B+\epsilon,
\]
Hence, the BPC algorithm of Algorithm~\ref{BPCAlgorithm} terminates in finite time and produces an optimal or $\epsilon$ optimal solution to the DNDP.
\end{proof}

\subsection{Column Generation}
Formulation \eqref{eq:station_placement_bi-level5} contains one path flow variable $h_{\pi}$ per feasible EV path $\pi\in\Pi (\mathbf{y})$. In realistic networks, $|\Pi (\mathbf{y})|$ is a large set. Following the approach used by \citep{ReyLevin2024}, we apply the CG method to efficiently solve the linear programming (LP) relaxation. We consider a restricted set of feasible paths \( \bar{\Pi} \subseteq \Pi (\mathbf{y}) \), and denote the corresponding restricted path sets as \( \bar{\Pi}_{rs} \subseteq \Pi_{rs} (\mathbf{y}) \) for all OD pairs \( (r,s) \in \mathcal{Z}^2 \). Problem \eqref{eq:station_placement_bi-level5} can be written as follows:
\begin{subequations}\allowdisplaybreaks \label{eq:rmp_ev_cg}
\begin{alignat}{2}
\min_{\mathbf{y}, \mathbf{x}, \mathbf{v}, \mathbf{h}, \pmb{\eta}, \pmb{\psi}} \quad
& -\sum_{l \in \mathcal{C}} p \cdot v_l + w \cdot \sum_{(r,s)} \left(D_{rs} - \sum_{\pi \in \bar{\Pi}_{rs}} h_{\pi}\right) && \label{eq:rmp_objective} \\[6pt]
\text{s.t.} \quad
& \sum_{l \in \mathcal{C}} C_l \cdot y_l \leq B, && \label{eq:rmp_budget} \\[6pt]
& 0 \leq y_l \leq 1, 
&& \quad \forall l \in \mathcal{C}
\label{eq:station_binary_constraint22} \\[6pt]
& v_l \leq M \cdot y_l, &&\quad \forall l \in \mathcal{C} \label{eq:rmp_bigM} \\[6pt]
& x_{ij} = \sum_{\pi \in \bar{\Pi}} \delta_{ij}^{\pi} h_{\pi}, &&\quad \forall (i,j) \in \mathcal{A} \label{eq:rmp_linkflow} \\[6pt]
& \sum_{\pi \in \bar{\Pi}_{rs}} h_{\pi} = d_{rs}(\mathbf{y}), &&\quad \forall (r,s) \in \mathcal{Z}^2 \label{eq:rmp_demandsatisfy} \\[6pt]
& v_l = \sum_{\pi \in \bar{\Pi}} \phi_l^{\pi} h_{\pi}, &&\quad \forall l \in \mathcal{C} \label{eq:rmp_chargingflow} \\[6pt]
& \int_{0}^{x_{ij}^1} t_{ij}(x)dx + t_{ij}(x^1)(x_{ij}-x_{ij}^1) \leq \eta_{ij}, && \quad \forall \mathbf{x}^1 \label{eq:rmp_link_oa} \\[6pt]
& \int_{0}^{v_l^1} t_l(v)dv + t_l(v^1)(v_l - v_l^1) \leq \psi_l, && \quad \forall \mathbf{v}^1 \label{eq:rmp_station_oa} \\[6pt]
& \sum_{(i,j)\in\mathcal{A}}\eta_{ij}+\sum_{l\in\mathcal{C}}\psi_l + \sum_{\pi \in \bar{\Pi}} c^{\pi} h_{\pi} \notag \\
& \leq \sum_{(i,j)\in\mathcal{A}}\int_{0}^{x_{ij}^\mathrm{f}}t_{ij}(x)dx
+\sum_{l\in\mathcal{C}}\int_{0}^{v_l^\mathrm{f}} t_l(v)dv + \sum_{\pi \in \bar{\Pi}} c^{\pi} h_{\pi}^\mathrm{f} +\sum_{l\in\mathcal{C}}M_l(1-y_l), &&\quad \forall (\mathbf{x}^{\mathrm{f}},\mathbf{v}^{\mathrm{f}}, \mathbf{h}^{\mathrm{f}}) \in \mathcal{F}\label{eq:rmp_vfcut} \\[6pt]
& h_{\pi}\geq 0, &&\quad \forall  \pi\in \bar{\Pi} \label{eq:rmp_bounds}
\end{alignat}
\end{subequations}
Since the objective function aims to maximize path flows \( h_{\pi} \) and charging flows \( v_l \), we reformulate the corresponding constraints as inequalities to expose the sign of their associated dual variables. Specifically, the demand satisfaction constraint \eqref{eq:rmp_demandsatisfy} becomes \( \sum_{\pi \in \bar{\Pi}_{rs}} h_{\pi} \leq d_{rs}(\mathbf{y}) \), the link flow definition \eqref{eq:rmp_linkflow} becomes \( \sum_{\pi \in \bar{\Pi}} \delta_{ij}^{\pi} h_{\pi} \leq x_{ij} \), and the charging flow definition \eqref{eq:rmp_chargingflow} becomes \( \sum_{\pi \in \bar{\Pi}} \phi_l^{\pi} h_{\pi} \geq v_l \). We derive dual variables associated with the constraints to determine the reduced cost of adding new paths to the LP. 
Denote by \( \sigma_{rs} \geq 0 \) the dual variables for the demand satisfaction constraints~\eqref{eq:rmp_demandsatisfy}, and by \( \zeta_{ij} \geq 0 \) the dual variables associated with link flow constraints~\eqref{eq:rmp_linkflow}. Given an OD pair \( (r,s) \in \mathcal{Z}^2 \) and a path \( \pi \in \Pi_{rs}^{\mathrm{f}} \), the reduced cost \( \bar{c}_{\pi} \) for the path flow variable \( h_{\pi} \) is computed as follows:
\begin{equation} \label{eq:reduced_cost}
\bar{c}_{\pi} = -\sigma_{rs} + \sum_{(i,j)\in\mathcal{A}} \delta_{ij}^{\pi}\zeta_{ij} - \sum_{l\in\mathcal{C}}\phi_l^{\pi}\gamma_l + c^{\pi},
\end{equation}
where \( \gamma_l \leq 0 \) represents the dual variables for the charging station flow constraints~\eqref{eq:rmp_chargingflow}, and \( c^{\pi} = \sum_{l \in \mathcal{C}} (c + \tau) \cdot \gamma_l^{\pi} \) is the path-dependent charging cost.
Thus, to identify new paths that have negative reduced cost (RC), one solves, for each OD pair \( (r,s) \), a battery constrained shortest path problem with link costs given by dual variables \( \zeta_{ij} \), station flow costs \( \gamma_l \), and charging costs \( (c + \tau) \cdot \gamma_l^{\pi} \). Since we have the expanded network we can solve the battery constrained shortest path problem (BCSPP) efficiently. If paths with negative reduced costs are found, these are introduced into the restricted path set \( \bar{\Pi} \) to refine the problem iteratively. This process continues until no paths with negative reduced cost are found. Algorithm \ref{alg:column_generation} shows the CG steps to iteratively solve problem \eqref{eq:rmp_ev_cg} using a limited set of paths and updates it by identifying and adding new paths with negative reduced cost. \\

\begin{algorithm}[H]
\DontPrintSemicolon
\caption{CG Algorithm}
\label{alg:column_generation}
Initialize \( \bar{\Pi} \gets \emptyset \), \( \text{RC} \gets -\infty \)\;
\While{\( \text{RC} < 0 \)}{
  %Solve problem \eqref{eq:rmp_ev_cg} over \( \bar{\Pi} \)\;
  Extract duals \( \sigma_{rs}, \zeta_{ij}, \gamma_l \) over \( \bar{\Pi} \)\;
  \For{\( (r,s) \in \mathcal{Z}^2 \)}{
    Build a BCSPP with arc‐costs \(\zeta_{ij}\), station‐flow‐costs \(-\gamma_l\), and charging‐costs \( (c + \tau) \cdot \gamma_l^{\pi}\)\;
    Solve BCSPP to get path \(\pi^\star_{rs}\) minimizing 
     \(\sum\zeta_{ij} - \sum\gamma_l + \sum (c + \tau) \cdot \gamma_l^{\pi}\).\;
     Compute its reduced cost:
     \[
       \bar{c}^\star \;=\; -\sigma_{rs}
         \;+\;\sum_{(i,j)\in\pi^\star_{rs}}\zeta_{ij}
         \;-\;\sum_{l\in\pi^\star_{rs}}\gamma_l
         \;+\; c^{\pi^\star_{rs}}
     \]
     \(\text{RC}_{\min}\gets\min(\text{RC}_{\min},\bar{c}^\star)\).\;
    \If{\( \bar{c}_{\pi^\star_{rs}} < 0 \)}{
      Add \( \pi^\star_{rs} \) to \( \bar{\Pi} \)\;
    }
  }
}
\end{algorithm}

\subsection{Battery State Network Expansion for EV Routing}
In our bi-level formulation, the LL problem captures user behavior in response to the network design decisions made by the UL planner. Specifically, users decide how to route their trips across the network to minimize their total travel and charging costs. This routing behavior directly affects traffic flow patterns and determines the demand at each charging station. Therefore, it is essential to model route choices accurately to reflect how users respond to infrastructure deployment. For EVs, route choice becomes of high importance due to battery limitations. Because of battery limitations, conventional shortest path assumptions are insufficient. The main challenge in EV routing lies in the BCSPP, where the shortest path from an origin to a destination must not only minimize travel time or distance, but also respect the vehicle’s battery limitations. In this problem, an EV begins its trip with a fully charged battery and must ensure that at every point along its route, the remaining battery is sufficient to reach the next node, or a charging station must be available en route to recharge. 

For a conventional vehicle, the shortest path between origin $r$ and destination $s$ is determined by minimizing travel time or distance without any operational constraints, since refueling is fast and gas stations are available and easy to find. However, for EVs, we must consider the feasible path set based on the charging stations. $\Pi_{(r,s)} (\mathbf{y})$ consists of all paths $\pi$ from $r$ to $s$ which are feasible for EV based on the charging configuration $\mathbf{y}$, which means the EV starts with full battery level $b$ at the origin $r$ and at every intermediate node $i$ along path $\pi$, the remaining battery level $b_i \geq 0$. The EV can either reach the next node directly or access a charging station to replenish its battery.
Therefore, a path $\pi = (i_1, i_2, ..., i_n)$ is feasible under $\mathbf{y}$ if and only if:
\begin{equation}
b_i - d(i, i+1) \geq 0 \text{ or }\text{ the user will charge at } i \quad \forall i \in {1, 2, ..., n-1}
\end{equation}
where $d(i, i+1)$ is the battery consumption to traverse from node $i$ to $i+1$. In the original network $\mathcal{G}$, finding the shortest path between the origin and the destination is constrained by battery level. If the shortest path length exceeds the available battery range, the EV cannot complete the trip and must charge en route. It is possible that the unconstrained shortest path does not belong to the feasible set. In addition to that, the feasible shortest path for EVs may have a longer travel time or distance compared to the unconstrained shortest path, or even no feasible path exists without charging.

Therefore, we cannot use the original network to find a shortest path that is also feasible for the EV. We must track battery levels alongside physical locations. We extend the traditional transportation network by using battery state as an additional dimension in the node representation.
The primary motivation behind this expansion is to model battery depletion along travel paths and to enable the inclusion of charging behavior at designated nodes. In the unexpanded graph, a node is a simple spatial location. In the expanded graph we consider a node to be a tuple $(i, b)$, where $i$ is the original node and $b$ is the remaining battery level upon arrival. This change allows us to enforce constraints on battery feasibility. A path is only considered feasible if, at every step, the EV has sufficient battery to reach the next node or can recharge at an available charging station along the route. An issue that comes up in the unexpanded graph is that the shortest path in terms of travel time may not be feasible for EVs due to insufficient battery capacity. As a result, EV users may instead use a longer path that passes through one or more charging stations, or that is otherwise feasible given their battery constraints. 

This conversion changes the battery constraint problem to an ordinary arc feasibility problem. A road arc only exists if the remaining battery is enough to traverse it. A zero length charging arc jumps from $(i, b)$ to $(i, b_{max})$ when a station is present. After this transformation, the EV routing problem is just a classical shortest path search on a graph with non‑negative costs, so we can run one standard labeling shortest path problem per origin to obtain the least cost feasible path to every destination. This is dramatically faster than solving a battery constrained shortest path MILP for each O‑D pair and lets us embed the routing step inside the LL UE in the BPC algorithm.
 
To construct the expanded network, we initialize the process by assigning all origin and charging nodes a full battery level. We then recursively explore all outgoing links from each node and battery pair. For each feasible traverse, a new node with the reduced battery level is created if it does not already exist, and a corresponding link is added to show the traversal. This continues until all reachable node and battery combinations are explored. In addition, charging actions are modeled through special charging links, which connect a node with a partial battery to the same physical node with a full battery level. These links are assigned zero physical length and a fixed base travel time corresponding to the charging duration per hour. This approach enables EV users to make charging decisions within the path optimization framework. Figure \ref{fig:combined} shows an original network and the expansion by assuming that the origin is node 1 and we have a charging station on node 3.
\begin{figure}[ht]
    \centering
    \begin{subfigure}[b]{0.3\linewidth}
        \centering
        \includegraphics[width=\linewidth]{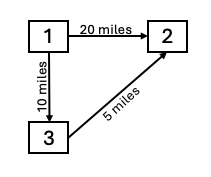}
        \caption{Original network}
        \label{fig:subfig-a}
    \end{subfigure}
    \hfill
    \begin{subfigure}[b]{0.3\linewidth}
        \centering
        \includegraphics[width=\linewidth]{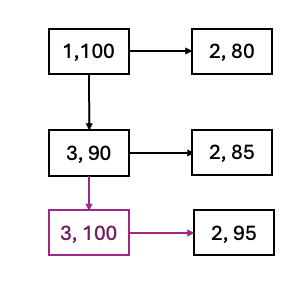}
        \caption{Battery-expanded network}
        \label{fig:subfig-b}
    \end{subfigure}
    \caption{Battery-State Network Expansion}
    \label{fig:combined}
\end{figure}

\section{Numerical Results}\label{Numerical Results}
We conducted numerical experiments on three different networks. The Eastern Massachusetts network, which comprises 74 zones, 74 nodes, and 248 links, the Anaheim network consists of 38 zones, 416 nodes, and 914 links, and the Barcelona network with 110 zones, 1020 nodes, and 2,512 links. In addition, various number of new charging station candidate locations were introduced. To accurately model EV routing behavior, we expanded the original network of Eastern Massachusetts by using battery state considerations. The expansion network resulted in an expanded network of 1,624 nodes and 2,612 links. The Anaheim network resulted in 4,015 nodes and 6,191 links in expanded form. The Barcelona network was expanded to 7,559 nodes and 9,169 links. This expanded representation allowed route choices that respect battery constraints and charging requirements. All computational experiments were performed on a MacBook Pro equipped with an Apple M4 Pro chip and 64 GB of unified memory, running macOS.

We derive the $t_{ij} (x_{ij})$, which is the link travel time in our model, from the BPR function, which takes the form:
\begin{equation}
    t_{ij} (x_{ij}) = t_{ij}^{0} \cdot \left( 1 + \alpha \left(\frac{x_{ij}}{C_{ij}}\right)^\beta\right)
\end{equation}
$\alpha$ and $\beta$ are shape parameters that can be calibrated to data. It is common to use $\alpha = 0.15$ and $\beta = 4$. $t_{ij}^{0}$ is the free flow travel time on link $(i, j)$, and $C_{ij}$ is the capacity of link $(i, j)$. This equation is used to derive the travel time of a link $(i, j)$. Similarly, the charging station congestion function $t_l(v_l)$ follows a BPR-style formulation:
\begin{equation}
    t_l(v_l) = t_l^{0} \cdot \left( \alpha \left(\frac{v_l}{K_l}\right)^\beta\right)
\end{equation}
where $K_l$ is the service capacity of station $l$ representing the maximum flow rate the station can handle efficiently. We use the same shape parameters $\alpha = 0.15$ and $\beta = 4$ for consistency with link travel time. The station capacity is determined by the investment cost through the relation $K_l = \kappa \cdot C_l$, where $\kappa$ is a capacity cost proportionality constant and $C_l$ is the installation cost at station $l$. Note that this equation captures the delay at the station due to congestion and queue, not the charging time. Charging time was captured on the third component of the objective function, which varied based on the arrival battery state. This formulation reflects the realistic assumption that stations with higher investment have greater capacity and can serve more vehicles with less congestion. The charging price is set to $c$ per unit time (e.g., dollars per hour), so the path-dependent charging cost is $c^{\pi} = \sum_{l \in \mathcal{C}} (c+ \tau) \cdot \gamma_l^{\pi}$, where $\gamma_l^{\pi}$ is the charging time for a vehicle on path $\pi$ at station $l$ based on its battery state upon arrival. This monetary cost captures both the direct electricity cost and the opportunity cost of time spent charging, ensuring that users with different arrival battery states incur different costs proportional to their required charging duration.

The charging time $\gamma_l^\pi$ is not an exogenous input but is endogenously determined through our battery-expanded network representation. In the expanded network, each node is represented as a tuple $(i, b)$, where $i$ is the physical location and $b$ is the remaining battery level. As vehicles traverse the network, the battery state is tracked at each node along the path. When a vehicle following path $\pi$ arrives at charging station $l$, the expanded network directly provides the arrival battery level $b_l^\pi$. The charging time is then calculated as:
\begin{equation}
    \gamma_l^\pi = (B^{\max} - b_l^\pi) \cdot \rho
\end{equation}
where $B^{\max}$ is the full battery capacity, $b_l^\pi$ is the arrival battery level at station $l$ for a vehicle following path $\pi$, and $\rho$ is the charging rate (time required to charge one unit of battery, e.g., hours per percentage point). Consequently, the charging cost for a vehicle on path $\pi$ at station $l$ becomes:
\begin{equation}
    c^{\pi} \cdot \gamma_l^\pi = (c + \tau) \cdot (B^{\max} - b_l^\pi) \cdot \rho
\end{equation}
This formulation ensures that vehicles arriving with depleted batteries spend more time charging and more cost and thus incur higher perceived charging costs than those arriving with higher charge levels, capturing the path dependent nature of charging costs in our model. The total charging cost $c^\pi$ that depends on path aggregates these costs across all charging stations visited along path $\pi$.

In the following numerical experiments, we assume that EVs depart from their origins with a full battery, i.e., $b = B^{\max}$. To ensure that the charging station placement problem remains practically relevant, we scale the link lengths in the network to represent long distance travel scenarios where en-route charging becomes necessary. This scaling reflects the intended application of our model for strategic infrastructure planning on highway networks or inter city corridors, where trip distances typically exceed the driving range of EVs.
It is important to note that the full battery assumption is not a fundamental limitation of our methodology. The battery-expanded network representation is flexible and can accommodate any initial battery state. Specifically, the model can be extended to handle: (i) a fixed partial initial battery for all users (e.g., $b = 0.7 \cdot B^{\max}$), representing conservative planning scenarios; (ii) origin destination specific initial battery states $b_{rs}^{\text{init}}$, which can show the heterogeneous trip characteristics; or (iii) stochastic initial battery states drawn from a distribution to reflect real world variability in departure conditions. These extensions require only modifying the initialization of the battery state in the expanded network construction, while the core algorithmic framework, including the BPC algorithm, value function cuts, and column generation remains entirely unchanged. This flexibility ensures that our approach is applicable to a wide range of practical scenarios beyond the specific assumptions used in our numerical experiments.

The model parameters were set as follows. The base station cost was set to $C = 200$ monetary units, with individual station costs randomly drawn from the range $[0.1C, 1.0C]$ using a uniform distribution which reflects heterogeneity in installation costs across candidate locations. The budget was calculated as a proportion of the total cost of all candidate stations, i.e., $B = B_{\text{prop}} \times \sum_{l \in \mathcal{L}} c_l$, where $B_{\text{prop}}$ represents the budget proportion varied across experiments. Station capacity was determined proportionally to investment cost using the relationship $K_l = \kappa \cdot c_l$, where the capacity-cost ratio $\kappa = 0.1$ ensures that stations with higher installation costs have correspondingly higher service rates. The charging price and revenue were set to be $p= c = 10$ per unit of flow, representing the revenue collected from users charging their vehicles. The value of time was set to $\tau = 1.0$ \$/minute (equivalent to \$60/hour), a commonly used value in transportation economics. The charging time per unit of battery was set to $\gamma = 1.0$ minute, meaning that recharging from empty to full (100 units) requires 100 minutes at the base rate. These parameter values were chosen to reflect realistic operating conditions while maintaining computational tractability.

Figures \ref{fig1} illustrate the convergence behavior and optimality gap percentage of our BPC algorithm on the Anaheim network, respectively. These graphs are for the case of 30 candidate locations for charging stations. The convergence graph demonstrates the progression of the LB and UB. The gap between the two bounds decreases steadily until convergence is achieved. Specifically, the optimal solution was identified after evaluating 33 nodes in the BPC tree, achieving a 0.00\% optimality gap, which shows the effectiveness of using value function cuts. Figure \ref{fig2} shows the number of paths in each BPC node, indicating the success of CG in identifying and adding useful paths.
\begin{figure}[htbp]
    \centering
    \includegraphics[width=0.99\linewidth]{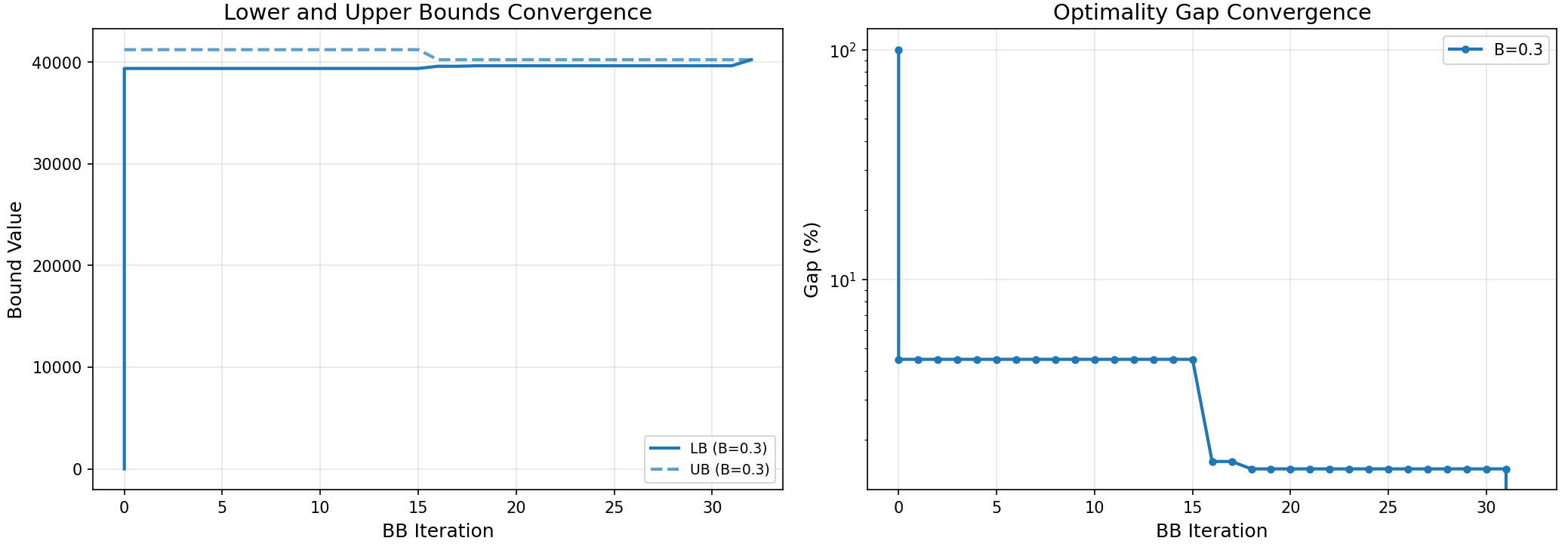}
    \caption{Convergence over BPC nodes}
    \label{fig1}
\end{figure}

\begin{figure}[htbp]
    \centering
    \includegraphics[width=0.55\linewidth]{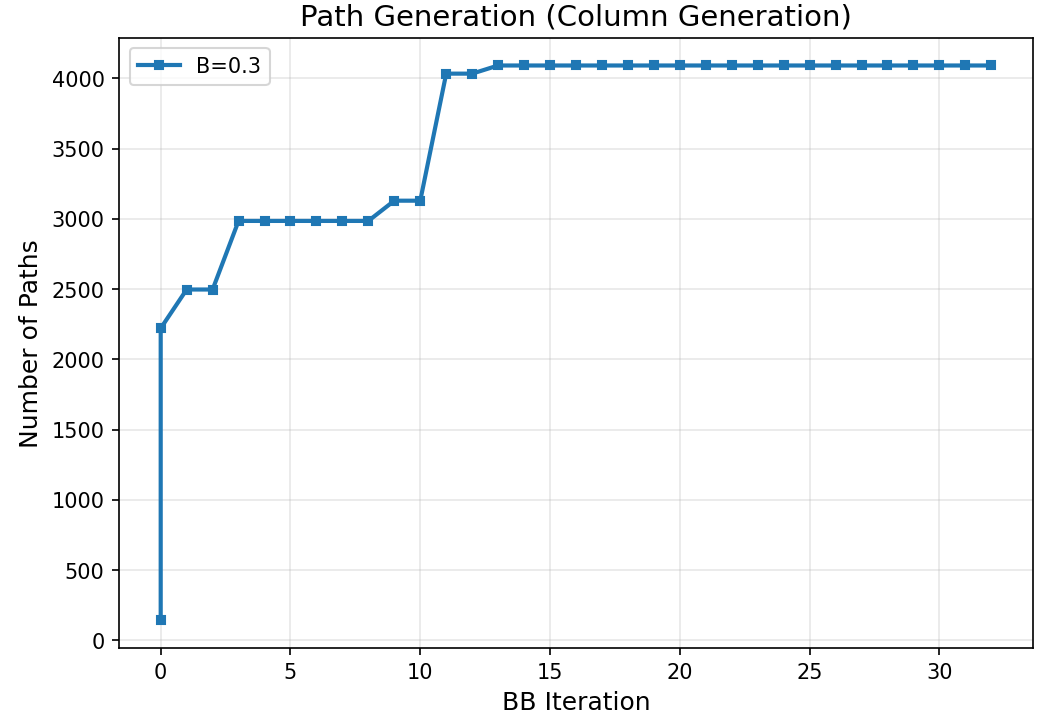}
    \caption{Number of Paths over the BPC nodes}
    \label{fig2}
\end{figure}

Table \ref{tab:results_all_networks} presents a detailed breakdown of computational performance across different problem sizes on all three networks, where ``Candidates'' shows the number of candidate locations. The ``TAP (s)'' column shows the time required to solve the MTAP, which determines traffic flow patterns given the current network configuration. The ``Pricing (s)'' represents the time needed to identify new paths to add to the HPR. ``RMP (s)'' shows the time spent solving HPR using CPLEX Solver \cite{cplex2009v12}. ``Total (s)'' indicates the total algorithm termination time, ``VF cuts'' represents the total number of value function cuts generated throughout the solution process, ``BB\_Nodes'' indicates the B\&B number of nodes until termination, and ``Comb.'' displays the total number of feasible combinations under the budget constraint. ``Gap (\%)'' shows the optimality gap at termination, ``Objective\_UB'' shows the upper bound objective function value, ``Revenue'' shows the charging revenue collected, and ``Unmet\_Demand'' shows the penalty for unserved demand.

\begin{sidewaystable}[htbp]
    \centering
    \scriptsize
    \setlength{\tabcolsep}{3pt}
    \caption{Summary of computational results across all networks}
    \label{tab:results_all_networks}
    \resizebox{\textheight}{!}{%
    \begin{tabular}{@{}ccc|cccccc|ccccc@{}}
        \hline
        Candidates & Comb. & BB\_Nodes & Gap (\%) & Objective\_UB & Revenue & Unmet\_Demand & Num\_Selected & Num\_Paths & Total time(s) & TAP (s) & RMP (s) & Pricing(s) & VF\_cuts \\
        \hline
        \multicolumn{14}{c}{\textbf{EasternMassachusetts ($B_{\text{Prop}} = 0.4$)}} \\
        \hline
        5 & 24 & 7 & 0.00 & 17142.60 & 59.16 & 172.62 & 2/5 & 568 & 18.76 & 2.84 & 1.49 & 0.54 & 2409 \\
        10 & 347 & 9 & 0.00 & 15753.97 & 239.00 & 162.25 & 4/10 & 2012 & 61.30 & 41.43 & 4.04 & 1.19 &  2478\\
        20 & 10480 & 2 & 0.60 & 15300.04 & 278.20 & 155.78 & 7/20 & 3835 & 80.00 & 59.92 & 4.12 & 0.63 &  2617\\
        25 & 54420 & 29 & 0.00 & 13954.52 & 428.28 & 145.77 & 10/25 & 6068 & 101.56 & 60.04 & 18.99 & 6.80 &  2630\\
        30 & 304786 & 47 & 0.00 & 13411.06 & 417.12 & 140.45 & 13/30 & 10904 & 235.99 & 150.16 & 48.08 & 20.31 &  2848\\
        \hline
        \multicolumn{14}{c}{\textbf{Anaheim ($B_{\text{Prop}} = 0.3$)}} \\
        \hline
        20 & 10480 & 19 & 0.47 & 40225.16 & 79.26 & 409.24 & 6/20 & 1968 & 281.54 & 47.34 & 174.98 & 5.23 & 5599 \\
        25 & 54420 & 17 & 0.75 & 40031.32 & 23.30 & 411.95 & 7/25 & 3672 & 427.32 & 71.24 & 258.16 & 8.61 & 6983 \\
        30 & 304786 & 33 & 0.00 & 40225.16 & 95.67 & 488.66 & 8/30 & 4091 & 779.44 & 124.97 & 535.82 & 16.70 & 7076 \\
        40 & > $10^6$ & 93 & 0.00 & 39885.12 & 181.12 & 406.82 & 12/40 & 16174 & 3169.09 & 276.80 & 2521.39 & 227.03 & 7908 \\
        \hline
        \multicolumn{14}{c}{\textbf{Barcelona ($B_{\text{Prop}} = 0.3$)}} \\
        \hline
        10 & 1024 & 9 & 0.60 & 36245.05 & 27.88 & 360.56 & 3/10 & 141 & 856.27 & 189.16 & 383.52 & 4.84 & 9303 \\
        20 & 54420 & 28 & 0.15 & 35990.40 & 52.95 & 359.90 & 7/20 & 199 & 1674.98 & 236.21 & 1106.95 & 12.18 & 9435 \\
        30 & 304786 & 25 & 0.77 & 35985.51 & 173.52 & 358.82 & 9/30 & 785 & 2325.35 & 294.87 & 1641.82 & 17.06 & 9723 \\
        40 & > $10^6$ & 101 & 0.91 & 35506.57 & 299.43 & 354.82 & 13/40 & 6000 & 10530.83 & 395.28 & 9597.46 & 112.89 & 9928 \\
        \hline
    \end{tabular}
}
\end{sidewaystable}
For the Eastern Massachusetts network, the algorithm achieves 0.00\% optimality gaps for most instances, with only the 20-candidate instance terminating with a small gap of 0.60\%. The largest instance (30 candidates, over 300,000 feasible combinations) is solved to optimality in 235.99 seconds. The results reveal that the TAP component constitutes the primary computational bottleneck, growing from 2.84 seconds for 5 candidates to 150.16 seconds for 30 candidates. The RMP and Pricing (CG) components also show considerable growth with problem size. The algorithm successfully identifies optimal solutions with up to 13 out of 30 stations selected.

For the Anaheim network, the algorithm solves all instances, with the two smaller instances (20 and 25 candidates) terminating with gaps of 0.47\% and 0.75\%, respectively, while the larger instances (30 and 40 candidates) achieve 0.00\% gaps. The largest Anaheim instance with 40 candidates required 3169.09 seconds and explored 93 B\&B nodes, selecting 12 out of 40 candidate stations. The most time-consuming component for this network is the RMP, consuming 2521.39 seconds out of the total runtime for the 40-candidate instance. This means that the number of paths added to the network is high, increasing the complexity of the master problem.

For the Barcelona network, which is the largest network tested, the algorithm maintains optimality gaps below 1\% across all instances. The most challenging instance with 40 candidates required 10530.83 seconds (approximately 2.9 hours) and explored 101 B\&B nodes, achieving a 0.91\% gap while selecting 13 out of 40 stations. The RMP dominates the computation time, accounting for 9597.46 seconds of the total runtime. This indicates that employing faster TAP algorithms and more efficient master problem solvers could considerably improve overall performance.

\subsection{Metaheuristic Benchmarks: Genetic Algorithms}
To benchmark the exact BPC method, we implement two population-based GA with different levels of complexity. We selected GAs as the benchmark heuristic because they are among the most widely used metaheuristics for network design problems in the transportation literature \cite{gaya, tomma, MIANDOABCHI20111041, juan, FARAHANI2013281, POORZAHEDY2007578, pattnaik}. First, a lightweight Basic GA with a small population and generation-based stopping, and a more complete Full GA with additional diversity mechanisms and a time-limit stopping rule. Both algorithms optimize the same upper-level decision vector and evaluate each candidate solution by solving the same lower-level UE model, so performance differences across methods are directly attributable to the search strategy rather than to model approximations.

In both heuristics, a candidate solution is encoded as a binary vector $\mathbf{y}=(y_1,\dots,y_{|\mathcal{C}|})$, where $y_l=1$ indicates that station $l\in\mathcal{C}$ is built and $y_l=0$ otherwise. For any $\mathbf{y}$, the fitness value is obtained by solving $\mathrm{MTAP}(\mathbf{y})$ with Frank Wolfe for TAP and evaluating the resulting UL objective.
\begin{equation}
f(\mathbf{y})=-\sum_{l\in\mathcal{C}}p \cdot v_l+ w \cdot \sum_{(r,s)}\left(D_{rs}-\sum_{\pi\in\Pi_{rs}(\mathbf{y})}h_\pi\right).
\end{equation}
Because each fitness evaluation requires solving a UE subproblem, if a binary vector has already been evaluated, its stored objective value is reused, avoiding redundant solving UE. Budget feasibility is enforced throughout the search. Any candidate violating
\begin{equation}
\sum_{l\in\mathcal{C}} C_l \cdot y_l \le B
\end{equation}
is repaired by iteratively deactivating the currently selected station with the largest installation cost until the budget constraint is satisfied.

\subsubsection{Basic Genetic Algorithm (GA)}
The Basic GA is designed as a low-overhead population-based baseline that keeps only the core GA components, random feasible initialization, tournament selection, one-point crossover, bit-flip mutation, single elitism, and budget repair. The population is small by design to limit the number of expensive UE evaluations. No adaptive mechanisms (immigrant injection, stagnation-triggered mutation scaling, or greedy seeding) are used. The algorithm terminates when either the maximum number of generations $G_{\max}$ is reached, or no improvement is found for $G_{\mathrm{stop}}$ consecutive generations. Since the stopping criterion is generation-based rather than time-based, the total number of UE evaluations is predictable and bounded by $N_{\text{pop}} \times G_{\max}$.

\begin{algorithm}[H]
\DontPrintSemicolon
\caption{Basic GA}
\label{alg:basic_ga}
Initialize a population of $N_{\text{pop}}$ random budget-feasible chromosomes\;
Evaluate all individuals by solving $\mathrm{MTAP}(\mathbf{y})$ (or cache lookup); record best incumbent $\mathbf{y}^*$\;
Set no-improvement counter to zero\;
\For{$g = 1, 2, \dots, G_{\max}$}{
    Copy elite individual to next generation\;
    \While{new population size $< N_{\text{pop}}$}{
        Select two parents by tournament selection\;
        Apply one-point crossover to produce two offspring\;
        Apply bit-flip mutation to each offspring\;
        Repair offspring if budget-infeasible\;
        Evaluate offspring via $\mathrm{MTAP}(\mathbf{y})$ (or cache lookup)\;
        Add offspring to new population\;
    }
    \eIf{new best found}{
        Update $\mathbf{y}^*$; reset no-improvement counter\;
    }{
        Increment no-improvement counter\;
    }
    \If{no-improvement counter $\ge G_{\mathrm{stop}}$}{
        \textbf{break}\;
    }
}
Return $\mathbf{y}^*$ and $f(\mathbf{y}^*)$\;
\end{algorithm}

\subsubsection{Full Population-Based Genetic Algorithm (Full GA)}
The Full GA uses a richer exploration mechanism through population search. The initial population of size $N_{\text{pop}}$ is constructed using a combination of seeded and random individuals to promote both diversity and early exploitation of promising regions:
\begin{enumerate}
    \item A null solution with no stations built ($\mathbf{y} = \mathbf{0}$), serving as a baseline.
    \item A greedy-cheapest solution, constructed by iteratively adding the least expensive candidate stations until the budget $B$ is exhausted.
    \item A greedy-expensive solution, constructed by prioritizing the most expensive (highest capacity) stations first.
    \item The remaining individuals are generated by random budget-feasible construction: candidate locations are shuffled randomly, and stations are activated sequentially, skipping any that would violate the budget constraint $\sum_{l \in \mathcal{C}} C_l \cdot y_l \leq B$.
\end{enumerate}

Parents are selected via tournament selection with tournament size $k$. For each tournament, $k$ individuals are drawn uniformly at random from the current population, and the individual with the best (lowest) fitness is selected as the parent. Offspring are generated using uniform crossover with probability $p_c$. For each gene position $l$, the alleles of the two parents are swapped with probability 0.5, producing two children. If either child violates the budget constraint, a repair operator is applied. Each gene in the offspring undergoes an independent bit-flip mutation with probability $p_m$. If any gene is flipped, the resulting chromosome is repaired if necessary. So budget-infeasible chromosomes are repaired by iteratively deactivating the most expensive active station (i.e., setting $y_l = 0$ for the $l$ with the largest $C_l$ among active stations) until feasibility is restored.

The top $n_e$ individuals from the current generation are carried directly into the next generation without modification (elitism), ensuring that the best-known solutions are never lost. To prevent premature convergence, the algorithm monitors the number of consecutive generations without improvement (stagnation). When stagnation exceeds a threshold of 10 generations, the mutation rate is adaptively doubled (capped at 0.5) to increase exploration. If stagnation persists beyond 15 generations, a set of random immigrants (10\% of the population) replaces the worst individuals to reinject diversity.
The Full GA also runs under a time limit and returns the best incumbent if the run is interrupted. Relative to the Basic GA, this variant has higher per-generation overhead and typically requires evaluating more candidate solutions; however, it can explore the feasible set more broadly when computational time is sufficient.

\begin{algorithm}[H]
\DontPrintSemicolon
\caption{Full GA (population-based)}
\label{alg:full_ga}
Build initial feasible population with seeded and random chromosomes\;
Evaluate all individuals by solving $\mathrm{MTAP}(\mathbf{y})$ and record best incumbent\;
\For{$g=1,2,\dots,G_{\max}$}{
\If{elapsed time $\ge T_{\max}$}{break and report incumbent\;}
Copy elite individuals to next generation\;
\While{new population size $< N_{\text{pop}}$}{
Select two parents by tournament selection\;
Apply uniform crossover to create offspring\;
Apply bit-flip mutation to offspring\;
Repair offspring if budget-infeasible\;
Evaluate offspring via $\mathrm{MTAP}(\mathbf{y})$ (or cache lookup)\;
Add offspring to new population\;
}
Update incumbent from the new population\;
\If{stagnation persists}{inject random feasible immigrants\;}
}
Return incumbent $\mathbf{y}^*$ and $f(\mathbf{y}^*)$\;
\end{algorithm}
\subsubsection{Parameterization Used in Experiments}
In the Full GA experiments we use $N_{\text{pop}}=40$, $G_{\max}=100$, crossover rate $p_c=0.8$, mutation rate $p_m=0.15$, tournament size $k=3$, elitism count $n_e=2$, and a wall-clock limit of $T_{\max}=3$ hours. In the Basic GA experiments we use $N_{\text{pop}}=12$, $G_{\max}=25$, crossover rate $p_c=0.8$, mutation rate $p_m=0.10$, tournament size $k=2$, elitism count $n_e=1$, and early stopping after $G_{\mathrm{stop}}=8$ generations without improvement. There is no wall-clock limit in the Basic GA; termination is defined entirely by generation count.

Because both genetic algorithms are stochastic, we evaluate each of them over five independent runs per instance, using the fixed random seeds $\{1,2,3,4,5\}$. The seed is set at the start of each run so that the entire search, population construction, tournament selection, crossover, mutation, and immigrant injection is fully reproducible, and the UE evaluation cache is cleared between runs so that the runs remain mutually independent. Accordingly, in Tables~\ref{tab:bench_EM}--\ref{tab:bench_barcelona} the GA objective and runtime are reported as the mean~$\pm$~standard deviation over these five runs, while the number of selected stations and the number of UE solves are reported as the mean rounded to the nearest integer. BPC is deterministic, so a single run is reported for it.

\subsubsection{Benchmark Results}
\label{sec:benchmark_results}
We now compare the three methods on identical problem instances across three networks we tested before. For each instance, all methods solve the same UL model, evaluate candidate solutions through the same LL which is UE subproblem on the same battery expanded network, and run on the same hardware. Consequently, performance differences are attributable solely to the search strategy.
The candidate station sets and installation costs used here are generated independently of those in Table~\ref{tab:results_all_networks}. Specifically, candidate locations are randomly sampled from the node set, and installation costs are drawn uniformly with a different random seed. As a result, the BPC objective values in Tables~\ref{tab:bench_EM}--\ref{tab:bench_barcelona} may differ from those reported earlier.
For each method we report the mean objective value over the five runs and its standard deviation (SD), the number of stations selected, the total runtime, the number of UE solves, and the runtime ratio of each heuristic relative to BPC. All instances use $B_{\mathrm{prop}} = 0.4$ for Eastern Massachusetts and $B_{\mathrm{prop}} = 0.3$ for Anaheim and Barcelona.
%%%%%%%%%
% ===================== EASTERN MASSACHUSETTS =====================
\begin{table}[ht]
\centering
\caption{Three-way comparison on the Eastern Massachusetts network
($B_{\mathrm{prop}}=0.4$)}
\label{tab:bench_EM}
\scriptsize
\resizebox{\textwidth}{!}{%
\begin{tabular}{r l r r r r r r r}
\hline
Candidates & Method & Objective\_UB (mean $\pm$ SD) & Num\_Sel & \#UE
& Total (s) (mean $\pm$ SD) & Gap (\%) & Heur.\ Gap (\%) & Slowdown \\
\hline
\multirow{3}{*}{5}
 & BPC      & 18643.55                  & 1/5 & 2
 & 19          & 0.00 & ---    & --- \\
 & Basic GA & 18643.55 $\pm$ 0.00       & 2/5 & 11
 & 38 $\pm$ 1  & ---  & 0.0000 & 2.0$\times$ \\
 & Full GA  & 18643.56 $\pm$ 0.01       & 2/5 & 10
 & 34 $\pm$ 1  & ---  & 0.0001 & 1.8$\times$ \\
\hline
\multirow{3}{*}{10}
 & BPC      & 17615.67                  & 4/10 & 2
 & 22           & 0.00 & ---    & --- \\
 & Basic GA & 17615.85 $\pm$ 0.19       & 4/10 & 56
 & 157 $\pm$ 8  & ---  & 0.0010 & 7.1$\times$ \\
 & Full GA  & 17615.68 $\pm$ 0.01       & 4/10 & 263
 & 908 $\pm$ 59 & ---  & 0.0001 & 41.3$\times$ \\
\hline
\multirow{3}{*}{20}
 & BPC      & 16752.06                  & 7/20 & 2
 & 302                & 0.00 & ---    & --- \\
 & Basic GA & 16754.73 $\pm$ 0.92       & 8/20 & 91
 & 12250 $\pm$ 595    & ---  & 0.0159 & 40.6$\times$ \\
 & Full GA  & 16754.76 $\pm$ 1.15       & 8/20 & 86
 & 11500 $\pm$ 830    & ---  & 0.0161 & 38.1$\times$ \\
\hline
\multirow{3}{*}{25}
 & BPC      & 15834.59                  & 9/25 & 2
 & 287                 & 0.00 & ---    & --- \\
 & Basic GA & 15839.86 $\pm$ 1.40       & 10/25 & 102
 & 13111 $\pm$ 1114   & ---  & 0.0333 & 45.7$\times$ \\
 & Full GA  & 15840.92 $\pm$ 1.58       & 10/25 & 84
 & 10585 $\pm$ 847    & ---  & 0.0400 & 36.9$\times$ \\
\hline
\multirow{3}{*}{30}
 & BPC      & 15325.50                  & 12/30 & 2
 & 340                 & 0.00 & ---    & --- \\
 & Basic GA & 15333.02 $\pm$ 1.89       & 11/30 & 102
 & 17255 $\pm$ 1619   & ---  & 0.0491 & 50.8$\times$ \\
 & Full GA  & 15334.47 $\pm$ 2.17       & 11/30 & 68
 & 10368 $\pm$ 1243   & ---  & 0.0585 & 30.5$\times$ \\
\hline
\end{tabular}}
\end{table}

% ===================== ANAHEIM =====================
\begin{table}[htbp]
\centering
\caption{Three-way comparison on the Anaheim network
($B_{\mathrm{prop}}=0.3$)}
\label{tab:bench_anaheim}
\scriptsize
\resizebox{\textwidth}{!}{%
\begin{tabular}{r l r r r r r r r}
\hline
Candidates & Method & Objective\_UB (mean $\pm$ SD) & Num\_Sel & \#UE
& Total (s) (mean $\pm$ SD) & Gap (\%) & Heur.\ Gap (\%) & Slowdown \\
\hline
\multirow{3}{*}{20}
 & BPC      & 39615.09                  & 7/20 & 22
 & 385                & 0.67 & ---    & --- \\
 & Basic GA & 39740.92 $\pm$ 34.29      & 8/20 & 113
 & 3686 $\pm$ 303    & ---  & 0.3176 & 9.6$\times$ \\
 & Full GA  & 39701.28 $\pm$ 21.81      & 7/20 & 305
 & 11142 $\pm$ 713   & ---  & 0.2176 & 28.9$\times$ \\
\hline
\multirow{3}{*}{25}
 & BPC      & 39785.68                  & 7/25 & 8
 & 658                & 0.13 & ---    & --- \\
 & Basic GA & 40010.38 $\pm$ 51.61      & 8/25 & 117
 & 2689 $\pm$ 213    & ---  & 0.5648 & 4.1$\times$ \\
 & Full GA  & 39919.73 $\pm$ 34.73      & 9/25 & 130
 & 11973 $\pm$ 599   & ---  & 0.3369 & 18.2$\times$ \\
\hline
\multirow{3}{*}{30}
 & BPC      & 39312.25                  & 10/30 & 7
 & 846                & 0.00 & ---    & --- \\
 & Basic GA & 39694.71 $\pm$ 77.40      & 10/30 & 137
 & 4932 $\pm$ 281    & ---  & 0.9729 & 5.8$\times$ \\
 & Full GA  & 39514.10 $\pm$ 48.66      & 10/30 & 88
 & 12053 $\pm$ 944   & ---  & 0.5134 & 14.2$\times$ \\
\hline
\multirow{3}{*}{40}
 & BPC      & 38855.51                   & 12/40 & 4
 & 3522                  & 0.00 & ---    & --- \\
 & Basic GA & 39767.80 $\pm$ 136.81      & 13/40 & 236
 & 11348 $\pm$ 889      & ---  & 2.3479 & 3.2$\times$ \\
 & Full GA  & 39218.34 $\pm$ 84.81       & 14/40 & 43
 & 15023 $\pm$ 1294     & ---  & 0.9338 & 4.3$\times$ \\
\hline
\end{tabular}}
\end{table}

% ===================== BARCELONA =====================
\begin{table}[htbp]
\centering
\caption{Three-way comparison on the Barcelona network
($B_{\mathrm{prop}}=0.3$)}
\label{tab:bench_barcelona}
\scriptsize
\resizebox{\textwidth}{!}{%
\begin{tabular}{r l r r r r r r r}
\hline
Candidates & Method & Objective\_UB (mean $\pm$ SD) & Num\_Sel & \#UE
& Total (s) (mean $\pm$ SD) & Gap (\%) & Heur.\ Gap (\%) & Slowdown \\
\hline
\multirow{3}{*}{10}
 & BPC      & 36215.01                  & 3/10 & 2
 & 317               & 0.52 & ---    & --- \\
 & Basic GA & 36230.38 $\pm$ 5.35       & 4/10 & 25
 & 348 $\pm$ 16     & ---  & 0.0424 & 1.0$\times$ \\
 & Full GA  & 36229.18 $\pm$ 3.54       & 3/10 & 114
 & 1488 $\pm$ 69    & ---  & 0.0391 & 4.7$\times$ \\
\hline
\multirow{3}{*}{20}
 & BPC      & 35937.43                  & 7/20 & 2
 & 410                & 0.00 & ---    & --- \\
 & Basic GA & 36015.13 $\pm$ 63.66      & 8/20 & 128
 & 2718 $\pm$ 214    & ---  & 0.2162 & 6.6$\times$ \\
 & Full GA  & 35997.57 $\pm$ 34.02      & 7/20 & 340
 & 10924 $\pm$ 953   & ---  & 0.1674 & 26.6$\times$ \\
\hline
\multirow{3}{*}{30}
 & BPC      & 35930.58                  & 9/30 & 4
 & 1087               & 0.62 & ---    & --- \\
 & Basic GA & 36145.88 $\pm$ 49.22      & 10/30 & 94
 & 3340 $\pm$ 231    & ---  & 0.5992 & 3.1$\times$ \\
 & Full GA  & 36074.81 $\pm$ 31.00      & 10/30 & 292
 & 10382 $\pm$ 661   & ---  & 0.4014 & 9.6$\times$ \\
\hline
\multirow{3}{*}{40}
 & BPC      & 35273.82                  & 13/40 & 2
 & 959                  & 0.26 & ---    & --- \\
 & Basic GA & 35534.60 $\pm$ 57.27      & 15/40 & 245
 & 17623 $\pm$ 1111    & ---  & 0.7393 & 18.4$\times$ \\
 & Full GA  & 35713.20 $\pm$ 91.13      & 15/40 & 161
 & 10642 $\pm$ 1080    & ---  & 1.2456 & 11.1$\times$ \\
\hline
\end{tabular}}
\end{table}

On Eastern Massachusetts, both genetic algorithms remain within 0.06\% of the certified BPC objective on average. On Anaheim and Barcelona, BPC strictly improves upon both GAs on every instance, and the mean heuristic gap widens with problem size. The largest gaps occur at $|\mathcal{C}|=40$, where the basic GA averages 2.35\% above BPC on Anaheim and full GA 1.25\% above on Barcelona, despite exhausting its three-hour time limit in both cases. The growing standard deviation across seeds at larger $|\mathcal{C}|$ reflects the combinatorial growth of the feasible set.

A fundamental limitation of heuristic methods is the absence of a principled stopping criterion. Because GAs do not maintain a lower bound on the optimal objective, there is no way to determine during the search whether the current incumbent is close to the global optimum or far from it. In practice, this forces the user to either set an arbitrary generation limit (as in Basic GA) or impose a time budget (as in Full GA), neither of which provides any guarantee that the search has converged. To give the heuristics a fair opportunity, we allowed the Full GA a generous three-hour time limit and equipped it with adaptive mutation, immigrant injection, and greedy seeding. Despite this, the Full GA still terminates with gaps of up to 2.35\% on the larger instances. More importantly, the user has no way of knowing whether running the algorithm longer would close this gap or whether the incumbent is already near-optimal. BPC, by contrast, maintains both a lower bound and an upper bound throughout the search. When these bounds meet within the prescribed tolerance, the algorithm terminates with a certificate of global optimality. This self-certifying property means that BPC stops precisely when the answer is proven, which is not only more reliable but also faster. BPC requires between 2 and 22 UE solves across all instances, whereas the GAs require between 10 and 340.

The runtime advantage of BPC is most pronounced on the larger instances. On Eastern Massachusetts with $|\mathcal{C}| \geq 20$, BPC is 31--51$\times$ faster than both GAs. On Anaheim and Barcelona, BPC is 3--30$\times$ faster in all instances. The efficiency stems from the value function cuts, which progressively tighten the HPR so that lower bounds are established through linear programming rather than repeated UE evaluation.

Figure~\ref{fig:convergence_barcelona} illustrates this contrast on the Barcelona instance with $|\mathcal{C}|=20$ and $B_{\mathrm{prop}}=0.3$ on one single run with seed 1. The BPC upper bound (black line) and lower bound (red line) converge to the certified optimum within 409 seconds, at which point the algorithm terminates with a proven gap of 0.00\% (dotted vertical line). The Full GA (blue) and Basic GA (orange) continue searching well beyond this point, decreasing their incumbent objectives over successive generations, but neither reaches the BPC objective. At the time BPC has already terminated, both GAs are still far from their final values and have no mechanism to assess how much further improvement is possible. At BPC's termination (409 seconds), the Full GA incumbent is 36086.59, which is 0.42\% above BPC's certified optimum; by the time it exhausts its three-hour budget, it improves to a final objective of 35964.27, corresponding to a gap of 0.07\%. Even this small residual gap is only quantifiable because BPC provides the reference value; without it, the GA user would have no basis for evaluating solution quality.
\begin{figure}[htbp]
    \centering
    \includegraphics[width=\textwidth]{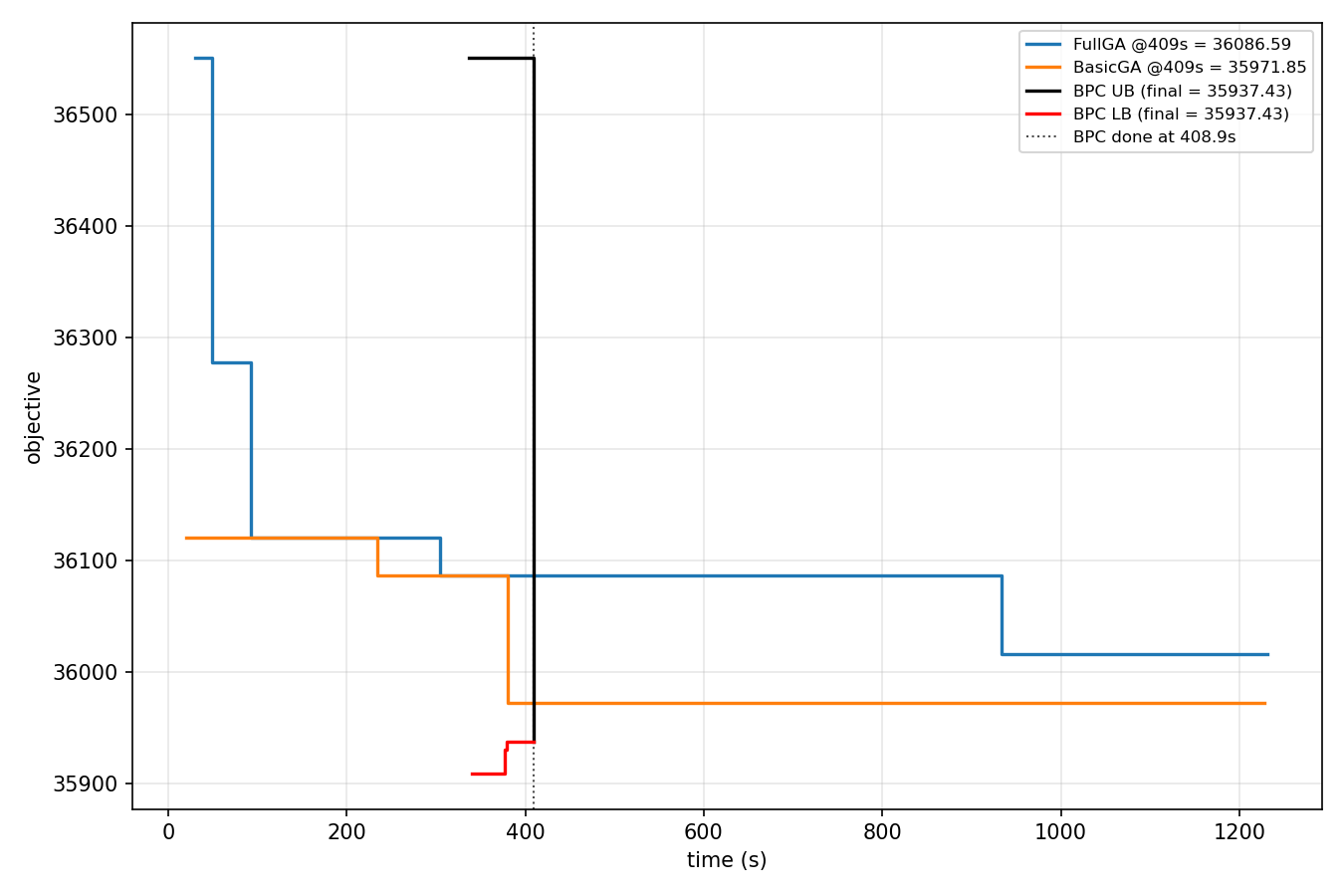}
    \caption{Convergence on Barcelona ($|\mathcal{C}|=20$, $B_{\mathrm{prop}}=0.3$). Solid lines show the best incumbent objective over time for Full GA (blue) and Basic GA (orange); the reported objective values are the incumbents at BPC's stopping time (409 seconds).} The black and red lines show the BPC upper bound and lower bound.
    \label{fig:convergence_barcelona}
\end{figure}

The GAs sometimes select more stations than BPC for the same or worse objective. On Anaheim with $|\mathcal{C}|=40$, both GAs select 13 and 14 stations while BPC selects 12 and achieves a better objective; on Barcelona with $|\mathcal{C}|=40$, both GAs select 15 versus BPC's 13. This over-building arises because the GAs lack the global view provided by the lower bound and cannot identify that additional stations do not improve the objective.

\section{Sensitivity Analysis}\label{sec:sensi}
In this bi-level optimization model for EV charging station placement, the budget parameter is the only factor that meaningfully influences the optimal station configuration. This occurs because the model exhibits captive demand in which EV users whose battery is insufficient to complete their trip must charge regardless of waiting time, congestion delays, or charging price. Consequently, both components of the objective function (maximizing revenue and minimizing unmet demand) are aligned rather than competing since building more stations always improves both metrics simultaneously. Changing the weight parameter ($w$) only rescales the penalty for unmet demand but does not alter which stations are selected, since all feasible configurations rank in the same order regardless of $w$. Similarly, varying the kappa parameter (capacity scaling factor) affects congestion delays but not station selection, because users have no alternative, and they will use the nearest available station even with high waiting times. The charging price also has no effect on station configuration because demand is inelastic; users cannot choose to forgo charging when their battery is depleted. Therefore, the budget constraint serves as the true limiting factor, and sensitivity analysis is most meaningful when examining how different budget levels affect which stations are selected and the resulting network performance.
\subsection{Budget Sensitivity Analysis}
To investigate the impact of budget availability on optimal charging station deployment, we conducted a sensitivity analysis on the Eastern Massachusetts network with 30 candidate locations, varying the budget proportion $B_{\text{prop}}$ from 0.1 to 0.9 (Table \ref{tab:sensitivity_budget}). 

\begin{table}[htbp]
    \centering
    \scriptsize
    \setlength{\tabcolsep}{3pt}
    \caption{Sensitivity analysis on budget proportion for Eastern Massachusetts (30 candidates)}
    \label{tab:sensitivity_budget}
    \resizebox{\textwidth}{!}{%
    \begin{tabular}{@{}cccccccccc@{}}
        \hline
        $B_{\text{prop}}$ & Gap (\%) & Num\_Selected & BB\_Nodes & Num\_Paths & Total(s) & TAP(s) & RMP(s) & Pricing(s) & Budget\_Util (\%) \\
        \hline
        0.1 & 0.00 & 4/30 & 11 & 4495 & 86.64 & 23.11 & 47.15 & 3.31 & 99.80 \\
        0.2 & 0.22 & 6/30 & 15 & 5473 & 116.48 & 26.29 & 71.90 & 4.76 & 97.61 \\
        0.3 & 0.00 & 9/30 & 27 & 4415 & 167.67 & 26.35 & 119.17 & 8.48 & 96.40 \\
        0.4 & 0.00 & 12/30 & 1 & 3794 & 52.12 & 27.57 & 10.29 & 0.50 & 91.45 \\
        0.5 & 0.00 & 12/30 & 1 & 3794 & 51.09 & 26.89 & 10.23 & 0.50 & 77.78 \\
        0.6 & 0.00 & 12/30 & 1 & 3794 & 51.46 & 27.08 & 10.15 & 0.48 & 65.60 \\
        0.7 & 0.00 & 12/30 & 1 & 3794 & 51.08 & 26.93 & 10.06 & 0.47 & 65.33 \\
        0.8 & 0.00 & 12/30 & 1 & 3794 & 51.40 & 27.17 & 10.10 & 0.47 & 49.03 \\
        0.9 & 0.00 & 12/30 & 1 & 3794 & 51.42 & 27.36 & 9.89 & 0.48 & 39.13 \\
        \hline
    \end{tabular}
    }%
\end{table}

\begin{figure}[htbp]
    \centering
    \includegraphics[width=0.99\linewidth]{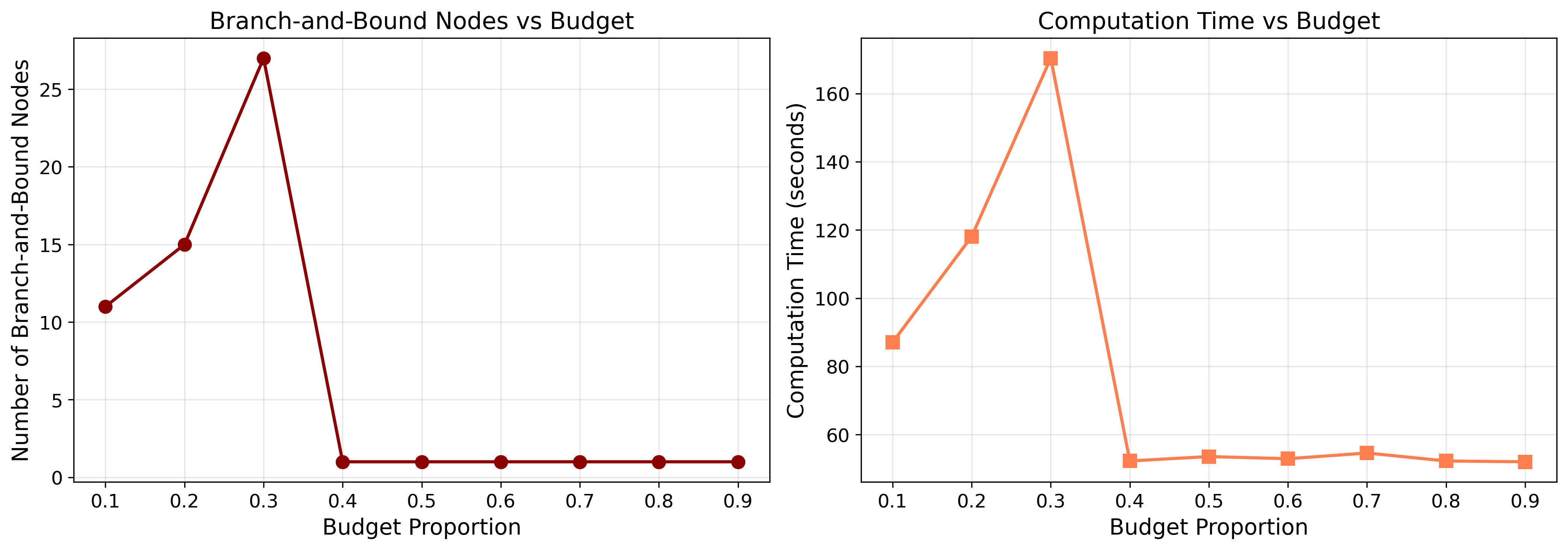}
    \caption{BPC tree nodes and computation time versus budget proportion}
    \label{fig:bb_time_budget}
\end{figure}

\begin{figure}[htbp]
    \centering
    \includegraphics[width=0.50\linewidth]{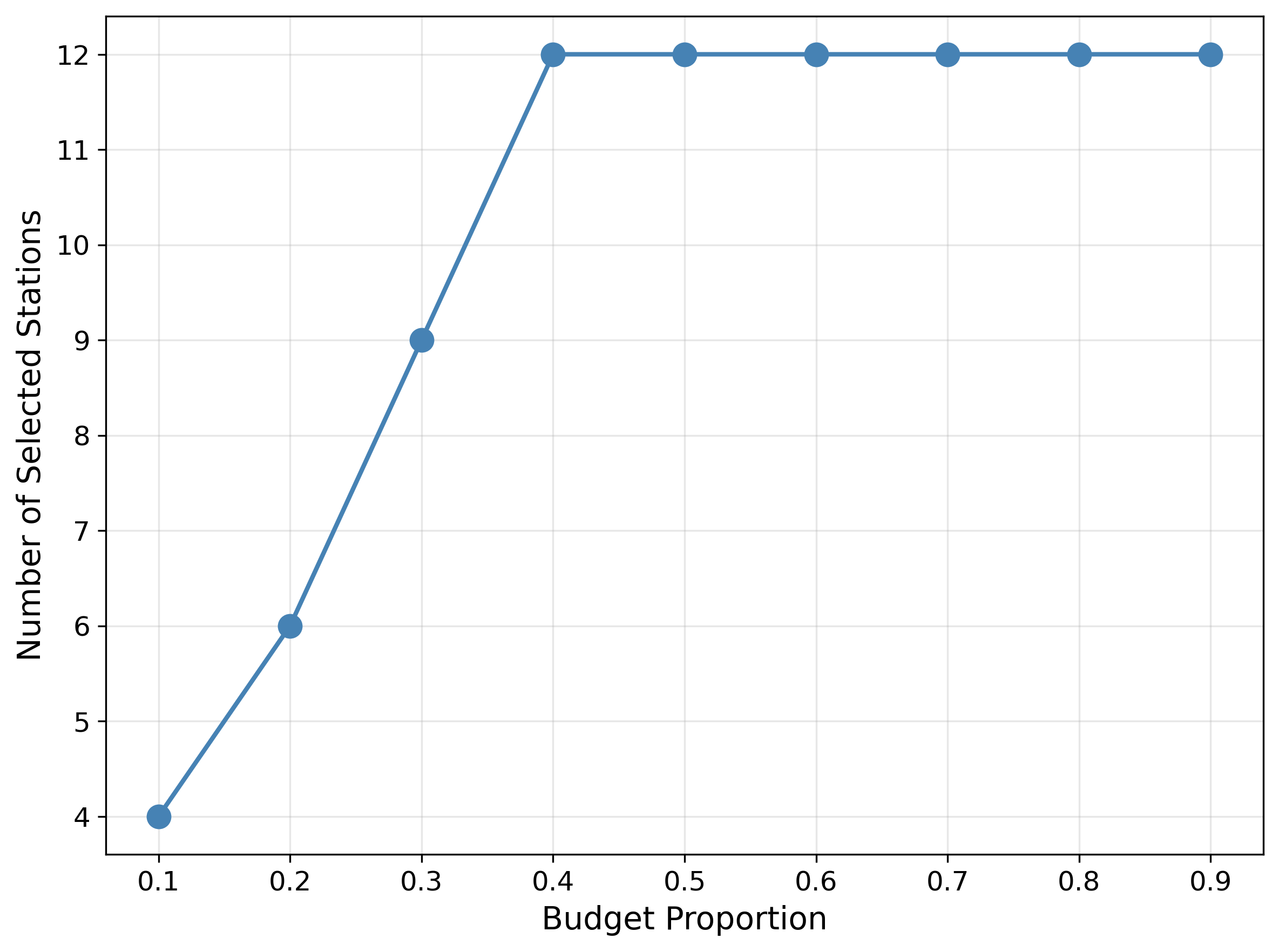}
    \caption{Number of selected stations versus budget proportion}
    \label{fig:selected_stations_budget}
\end{figure}

As shown in Figure \ref{fig:selected_stations_budget}, the number of selected stations increases monotonically from 4 at $B_{\text{prop}} = 0.1$ to 12 at $B_{\text{prop}} = 0.4$, after which it plateaus. This saturation indicates that 12 stations are sufficient to serve all feasible demand in the network, and additional budget beyond $B_{\text{prop}} = 0.4$ provides no further benefit in terms of station deployment. This finding is practically important for planners, as it identifies the minimum budget threshold beyond which additional investment yields no improvement in network coverage.

The computational behavior exhibits a non-monotonic pattern with respect to budget, as illustrated in Figure \ref{fig:bb_time_budget}. The number of BPC nodes peaks at 27 for $B_{\text{prop}} = 0.3$ and drops sharply to just 1 node for $B_{\text{prop}} \geq 0.4$. Correspondingly, the total computation time reaches its maximum of 167.67 seconds at $B_{\text{prop}} = 0.3$ and decreases to approximately 51 seconds for higher budget levels. This behavior can be explained by the combinatorial difficulty of the problem. When the budget is tight, the algorithm must explore more branches to determine the optimal subset of stations under binding budget constraints. At moderate budgets ($B_{\text{prop}} = 0.3$), the problem reaches peak difficulty because there are many near-optimal feasible combinations to evaluate. Once the budget is sufficiently large ($B_{\text{prop}} \geq 0.4$), the budget constraint becomes non-binding, and the algorithm can identify the optimal solution at the root node without branching.

The budget utilization decreases steadily from 99.80\% at $B_{\text{prop}} = 0.1$ to 39.13\% at $B_{\text{prop}} = 0.9$, confirming that higher budget levels lead to increasingly slack budget constraints. These results demonstrate that our BPC algorithm adapts efficiently to varying problem difficulty, requiring more computational effort when the problem is combinatorially challenging (tight budgets) and converging rapidly when the solution structure is simpler (loose budgets).

\subsection{Station Capacity Sensitivity Analysis}

We examine the effect of the capacity-cost proportionality constant $\kappa$ on charging station congestion and optimal deployment. Recall that the service capacity of each station is determined by $K_l = \kappa \cdot C_l$, where $C_l$ is the installation cost. Higher values of $\kappa$ yield greater station capacity, reducing congestion delays experienced by EV users. We vary $\kappa \in \{0.2, 0.5, 1.0, 1.5\}$ on the Eastern Massachusetts network with 20 candidate locations and $B_{\text{prop}} = 0.3$.

Table \ref{tab:kappa_capacity} summarizes the station capacities for representative stations at each $\kappa$ value, illustrating how low $\kappa$ values result in station capacities well below the expected flow levels.

\begin{table}[htbp]
    \centering
    \scriptsize
    \caption{Station capacity $K_l$ for selected stations at different $\kappa$ values}
    \label{tab:kappa_capacity}
    \begin{tabular}{@{}cccccc@{}}
        \hline
        Station & Cost ($C_l$) & $K$ ($\kappa$=0.2) & $K$ ($\kappa$=0.5) & $K$ ($\kappa$=1.0) & $K$ ($\kappa$=1.5) \\
        \hline
        231 & 196.90 & 3.96 & 9.90 & 19.80 & 29.70 \\
        226 & 183.75 & 3.70 & 9.24 & 18.48 & 27.72 \\
        182 & 112.03 & 2.25 & 5.63 & 11.27 & 16.90 \\
        187 & 92.89 & 1.87 & 4.67 & 9.34 & 14.01 \\
        179 & 66.61 & 1.34 & 3.35 & 6.70 & 10.05 \\
        \hline
    \end{tabular}
\end{table}

Consequently, small changes in $\kappa$ as presented in Table \ref{tab:sensitivity_kappa} can produce large differences in congestion when stations operate near or above capacity.
\begin{table}[htbp]
    \centering
    \scriptsize
    \caption{Sensitivity of station congestion to $\kappa$ on Eastern Massachusetts (20 candidates, $B_{\text{prop}} = 0.3$)}
    \label{tab:sensitivity_kappa}
    \begin{tabular}{@{}cccccc@{}}
        \hline
        $\kappa$ & Num\_Selected & Avg Delay (min) & Revenue & Gap (\%) & Total(s) \\
        \hline
        0.2 & 6/20 & 24.66 & 284.57 & 0.00 & 293.57 \\
        0.5 & 6/20 & 20.25 & 284.57 & 0.00 & 307.58 \\
        1.0 & 6/20 & 1.36 & 284.57 & 0.00 & 294.44 \\
        1.5 & 5/20 & 0.66 & 284.57 & 0.90 & 172.47 \\
        \hline
    \end{tabular}
\end{table}
The results reveal that $\kappa$ has a pronounced effect on station congestion while leaving the overall network design largely unchanged. For $\kappa = 0.2$, the average congestion delay across active stations is 24.66 minutes, reflecting severely constrained capacity where station flows substantially exceed the effective service rate. As $\kappa$ increases to 0.5, the delay decreases modestly to 20.25 minutes. However, at $\kappa = 1.0$, the delay drops dramatically to 1.36 minutes, representing a reduction of over 94\% compared to $\kappa = 0.2$. 

Notably, the same set of 6 stations is selected for $\kappa \in \{0.2, 0.5, 1.0\}$, with identical revenue of 284.57, indicating that the optimal station locations are robust to capacity assumptions within this range. At $\kappa = 1.5$, one fewer station is selected (5 instead of 6), suggesting that when individual stations have sufficiently high capacity, fewer facilities are needed to serve the same demand. The revenue remains unchanged at 284.57, as the same total flow is served regardless of the number of stations. The computation time also decreases at $\kappa = 1.5$ (172.47 seconds versus approximately 300 seconds for other values), consistent with the reduced combinatorial complexity of selecting fewer stations.

These findings have practical implications for infrastructure planning. The parameter $\kappa$ reflects the relationship between investment cost and service capacity, which depends on technology choices such as the number of chargers per station and their power output. Our results suggest that investing in higher-capacity stations (larger $\kappa$) can substantially reduce user waiting times without altering the optimal spatial distribution of stations, up to a threshold beyond which fewer but larger stations become preferable.

\section{Conclusion}\label{Conclusion}
In conclusion, this study successfully formulates and implements a bi-level optimization model for optimizing the placement of EV charging stations.  By applying a novel BPC algorithm enhanced with value function cuts and CG, we have achieved solutions within the reported $\epsilon$ optimality gap of global optimum, even for large scale networks such as Barcelona, which has 110 zones, 1,020 nodes, and 2,512 links.

Our methodology advances the state of the art in charging station location optimization by providing an exact algorithm with proven global optimality guarantees. While existing heuristic methods provide no optimality guarantees and often struggle with solution quality, our algorithm produces global or $\epsilon$ global optimal solutions with mathematical certainty.
Differences in formulation, objective structure, constraint design, and network representation across previous studies \cite{doi:10.1287/trsc.2021.0494, 10122631, 10034431,9690622} preclude direct head-to-head comparison, and we do not draw comparative performance claims from these heterogeneous results. Instead, we validate our algorithm through controlled benchmarks against two genetic algorithms on identical problem instances, where BPC consistently finds equal or better solutions in 3--50$\times$ less computation time across all three tested networks, while additionally providing a certificate of global optimality that no heuristic can offer. Specifically, our computational experiments demonstrate rapid convergence, with optimality gaps below 1\% across all instances and runtimes ranging from seconds to 2.9 hours on networks with up to 1,020 nodes and 2,512 links. The algorithm handles extensive scenarios involving thousands of candidate combinations, demonstrating that metropolitan-scale EV charging infrastructure planning is a tractable optimization task.

Future research can extend the congestion model to capture real-time dynamic queuing, such as real time station congestion. Furthermore, evaluating uncertainties in user behavior, demand forecasts, and technological advancements in battery capacities can enhance practical applicability. Recent work by \citet{RASHGISHISHVAN2025101761} shows that the spatial distribution of charging infrastructure considerably affects load profile smoothness and ramping needs, which suggests that coupling transportation network optimization with power system objectives could yield more holistic infrastructure planning. Additionally, the algorithm runtime can be improved by using faster algorithms like traffic assignment by paired alternative segments \citep{BARGERA20101022} instead of Frank Wolfe that we implement.

\section*{Acknowledgment}
The authors received no financial support for the research, authorship, and/or publication of this article.

\newpage

\bibliographystyle{trb}
\bibliography{trb_template}
\end{document}